\documentclass[12pt]{article}
\usepackage{geometry} 

\usepackage[utf8x]{inputenc} 
\usepackage{amsmath}
\usepackage{mathtools}
\usepackage{titling}
\usepackage{lipsum}
\usepackage{lscape}
\usepackage{amsfonts}
\usepackage{amssymb}
\usepackage{mathrsfs}
\usepackage{amscd}
\usepackage{physics}
\usepackage{amsthm}
\usepackage{subcaption}
\usepackage[dvips]{graphicx}
\usepackage[english]{babel}
\graphicspath{{noiseimages/}}

\usepackage{pst-node}
\usepackage{tikz-cd}

\usetikzlibrary{decorations.pathmorphing}
\begingroup
\catcode`\+\active\gdef+{\mathchar8235\nobreak\discretionary{}{\usefont{OT1}{cmr}{m}{n}\char43}{}}
\catcode`\-\active\gdef-{\mathchar8704\nobreak\discretionary{}{\usefont{OMS}{cmsy}{m}{n}\char0}{}}
\catcode`\=\active\gdef={\mathchar12349\nobreak\discretionary{}{\usefont{OT1}{cmr}{m}{n}\char61}{}}
\catcode`\<\active\gdef<{\mathchar"313C\nobreak\discretionary{}{\usefont{OML}{cmm}{m}{n}\char60}{}}
\catcode`\>\active\gdef>{\mathchar"313E\nobreak\discretionary{}{\usefont{OML}{cmm}{m}{n}\char62}{}}
\endgroup
\def\times{\mathchar8706\nobreak\discretionary{}{\usefont{OMS}{cmsy}{m}{n}\char2}{}}
\def\subset{\mathchar"321A\nobreak\discretionary{}{\usefont{OMS}{cmsy}{m}{n}\char26}{}}
\def\supset{\mathchar"321B\nobreak\discretionary{}{\usefont{OMS}{cmsy}{m}{n}\char27}{}}

\def\neq{\not=\nobreak\discretionary{}{\usefont{OMS}{cmsy}{m}{n}\char54\usefont{OT1}{cmr}{m}{n}\char61}{}}

\def\in{\mathchar"3232\nobreak\discretionary{}{\usefont{OMS}{cmsy}{m}{n}\char50}{}}
\def\to{\mathchar"3221\nobreak\discretionary{}{\usefont{OMS}{cmsy}{m}{n}\char33}{}}
\def\?#1{#1\nobreak\discretionary{}{\hbox{$\mathsurround=0pt #1$}}{}}

\title{The three-point Gaudin model and branched coverings of the Riemann sphere}

\author{Natalia Amburg
        \and
        Ilya Tolstukhin
        }

\usepackage{graphicx}
\newtheorem{theorem}{Theorem}
\newtheorem{remark}{Remark}
\newtheorem{corollary}{Corollary}
\newtheorem{hypothesis}{Hypothesis}
\newtheorem{lemma}{Lemma}
\newtheorem{observation}{Observation}
\newtheorem{proposition}{Proposition}
\newtheorem{definition}{Definition}

\begin{document}
\newpage
\maketitle


\begin{abstract}

We study the three-point quantum $\mathfrak{sl_2}$-Gaudin model. In this case the compactification of the parameter space is $\overline{M_{0,4}(\mathbb{C})}$, which is the Riemann sphere. We analyze sphere
coverings by the joint spectrum of the Gaudin Hamiltonians treating them as algebraic curves. We write equations of these curves as determinants of tridiagonal matrices and deduce some consequences regarding the geometric structure of the Gaudin coverings.
\end{abstract}

\tableofcontents

\section{Introduction}

The Gaudin spin chain model is a physical model related to the Lie algebra $\mathfrak{sl}_2$, which models the interactions of a number of particles with spin on a line. It was introduced in the paper \cite{GD} and was then generalized to the case of any semisimple Lie algebra.  The Gaudin Hamiltonians are $n$ commuting operators, which depend on pairwise distinct complex parameters $z_1,z_2,...,z_n$ and act on the $n$-fold tensor product of irreducible representations of a given semisimple Lie algebra. One of the main problems of the Gaudin model is to understand the joint spectrum of these operators and how it changes when the parameters vary. 

In his papers \cite{R1} and \cite{R2} Rybnikov showed that parametrization of the Gaudin algebras by pairwise distinct complex numbers $z_1,z_2,...,z_n$ can be naturally extended to the Deligne-Mumford compactification $\overline{M_{0,n+1}(\mathbb{C})}$ of genus $0$ curves with $n+1$ marked points. In the same works the coverings of $\overline{M_{0,n+1}(\mathbb{C})}$ by the spectrum of the Gaudin algebras are studied. 

In the present work, we consider the three-point Gaudin model, so that $n=3$. In this case, the Deligne-Mumford compactification of parameter space is $\overline{M_{0,4}(\mathbb{C})}$, which is the Riemann sphere. The coverings of $\overline{M_{0,4}(\mathbb{C})}$ by the joint spectrum of the Gaudin Hamiltonians are the branched coverings of the sphere, or seen differently, the algebraic curves. Natural questions arise about the equations of these curves and the genus of irreducible components. We are also interested in the branch points and the monodromy group of these coverings. In the present work, we are exploring approaches to these problems.

In section $\ref{degrees}$ we obtain formulas on degrees of the Gaudin coverings. In section $\ref{Shapovalov}$ we study the Shapovalov form and construct a family of bases of the space of singular vectors of a given weight, which is orthonormal with respect to it. In section $\ref{Shbases}$ we study the action of the Gaudin Hamiltonians on these bases, which we call Shapovalov bases. Each basis of this family remarkably does not depend on the point in the parameter space, and the action of the Gaudin Hamiltonians in it is given by a tridiagonal matrix. Depending on the three given irreducible $\mathfrak{sl}_2$-modules and parameter $u \in \overline{M_{0,4}(\mathbb{C})}$ on the Riemann sphere, these tridiagonal matrices are written explicitly. In section $\ref{Tridiagonal}$ we derive some consequences and describe the geometrical structure of the Gaudin coverings in terms of these matrices. In section $\ref{Hypotheses}$ we provide some observations and new interesting hypotheses regarding the structure of the Gaudin coverings and its monodromy. Section $\ref{Ornaments}$ is devoted to concrete examples of the Gaudin algebraic curves and their branch points, which we were able to calculate using a tridiagonal form of the Gaudin Hamiltonians in the Shapovalov bases.

\section{The Gaudin model and coverings of $\overline{M_{0, n+1}}(\mathbb{C})$}

We follow \cite{R1} in our exposition of the Gaudin model. Let $\mathfrak{g}$ be a semisimple complex Lie algebra (we will be mostly interested in the case $ \mathfrak{g}  = \mathfrak{sl}_2$) and $U(\mathfrak{g})$ be its universal enveloping algebra. 
Let ${x_{a}}, \;  a = 1,...,dim \: \mathfrak{g}$ be an orthonormal basis of $\mathfrak{g}$ with respect to the Killing form and ${x^{a}}, \; a = 1,...,dim \: \mathfrak{g}$ be its dual basis. Let us fix a positive integer $n$. For any $x \in U(\mathfrak{g})$ we  define  $x(i) := 1 \otimes ... \otimes 1  \otimes x \otimes 1 \otimes ... \otimes 1 \in  U(\mathfrak{g})^{\otimes n}$, where $x$ stands on $i$-th place. We also denote by $V_{\mu}$ the irreducible representation of $\mathfrak{g}$ of highest weight $\mu$. For any collection of dominant integral weights  $\underline{\lambda} = (\lambda_1, ..., \lambda_n)$ we set by definition  $V_{\underline{\lambda}} = V_{\lambda_1} \otimes ... \otimes  V_{\lambda_n} $.

Let us now fix a collection of pairwise distinct complex numbers  $\underline{z} = (z_1, ..., z_n)$. The Gaudin Hamiltonians are the following commuting operators acting on the space $V_{\underline{\lambda}}$: 
$$ H_i = \sum_{j \neq i} \sum_{a = 1}^{dim \: \mathfrak{g}} \; \frac{x_{a}(i) \; x^{a} {(j)} }{z_i - z_j } $$

and one of the main problems of the Gaudin model is to understand their simultaneous diagonalization.

The Gaudin Hamiltonians can be viewed as the elements of the space  $U(\mathfrak{g})^{\otimes n}$, and in the paper \cite{FFR} an existence of a large commutative subalgebra $ \mathcal{A}(\underline{z}) = \mathcal{A}(z_1,...,z_n) \subset U(\mathfrak{g})^{\otimes n}$, which contains all $H_i$, was proved. This subalgebra commutes with the diagonal action of $\mathfrak{g}$ on $U(\mathfrak{g})^{\otimes n}$ and in fact is the maximal commutative subalgebra in $[U(\mathfrak{g})^{\otimes n}]^{\mathfrak{g}}$ with this property. We call $ \mathcal{A}(\underline{z})$  the Gaudin algebra. Detailed construction of this algebra is in section  7.10 of the paper \cite{HKRW}.

 In the case  $ \mathfrak{g}  = \mathfrak{sl}_2$ subalgebra  $ \mathcal{A}(\underline{z}) \subset U(\mathfrak{g})^{\otimes n}$  is generated by the elements  $H_i$ and the center
 $U(\mathfrak{g})^{\otimes n}$, see \cite{R1}. More explicitly: let $e, f, h$ be the standard basis of $\mathfrak{sl}_2$, so that the dual basis is $f, e, h/2$. Then $ \mathcal{A}(\underline{z})$ is generated by  $ H_i = \sum_{j \neq i}  \; \frac{e(i) f(j) + f(i) e(j) + \frac{1}{2} h(i) h(j) }{z_i - z_j } $ and  $C_i = e(i)f(i) + f(i) e(i) + 1/2 h(i) h(i), \; \; i = 1, ..., n$ as an algebra. The only algebraic relation between the generators is  $\sum_{i}^{n} H_i = 0$, so the Gaudin algebra, in this case, is a polynomial algebra in $2n-1$ variables.

  The Gaudin algebras are parametrized by a non-compact algebraic variety $M_{0, n+1}(\mathbb{C})$ of the configurations of $n$ pairwise distinct points on a complex line up to an affine transformation. In section 3 of the paper \cite{R1} it was proved that this parametrization can be naturally extended to a sheaf of commutative algebras $\mathcal{A}$ on $\overline{M_{0, n+1}}(\mathbb{C})$, -- the Deligne-Mumford compactification of rational curves with $n+1$ marked points. The representation of $V_{\underline{ \lambda }} $ the universal enveloping algebra $U(\mathfrak{g})^{\otimes n}$ (in the notation of the previous chapter)  induces the morphism of sheaves of commutative algebras $$ \mathcal{A} \to \textit{End}(V_{\underline{\lambda}}) $$ where $\textit{End}(V_{\underline{\lambda}})$ is the constant sheaf with the stalk $\mathfrak{gl}(V_{\underline{\lambda}})$. We denote the image of this morphism as  $\mathcal{A}_{V_{\underline{\lambda}}}$. The stalks of the sheaf  $\mathcal{A}_{V_{\underline{\lambda}}}$ on the open subvariety ${M_{0, n+1}}(\mathbb{C}) $ are the commutative subalgebras of $\mathfrak{gl}(V_{\underline{\lambda}})$ generated by the Gaudin Hamiltonians. It is natural to consider their spectrum, and in the same work \cite{R1} it was proved an existence of a branched covering $$\pi_{V_{\underline{\lambda}}}: C_{V_{\underline{\lambda}}} \to \overline{M_{0, n+1}}(\mathbb{C})$$
 
 such that preimage of any point  $\underline{z} \in {M_{0, n+1}}(\mathbb{C})$ is the joint spectrum of the Gaudin Hamiltonians corresponding to the parameters $\underline{z}$.

\section{Bethe ansatz equations}

From now on, we consider the case of $\mathfrak{g} = \mathfrak{sl}_2$, so $\underline{\lambda} =(m_1,m_2, ..., m_n)$ is just a collection of non-negative integers. Let  $e, f, h$ be the standard basis of $\mathfrak{sl}_2$.
\begin{observation}\label{obs}

The Gaudin Hamiltonians commute with the diagonal action of $\mathfrak{sl_2}$ on $U(\mathfrak{sl_2})^{\otimes n}$:
$$ x H_i = H_i x \quad \forall x \in \mathfrak{sl_2} $$

which implies that they preserve singular (highest with respect to the diagonal action) vectors and weight spaces. Therefore, it suffices to diagonalize the Gaudin Hamiltonians on the space $(V_{\underline{\lambda}})^{sing}_{\mu}$  of highest vectors of weight $\mu$ for all possible $\mu \in \mathbb{Z}_{\geq 0}$. 
\end{observation}

Let us now describe what Bethe ansatz equations are. 
We search for eigenvectors of Gaudin Hamiltonians in a specific form. For any complex number $w \neq z_i$ let
$$ f(\omega) = \sum_{i =1}^{n} \frac{f(i)}{\omega - z_i} \in  U(\mathfrak{sl}_2)^{\otimes n} $$
Take the tensor product of the highest weight vectors in each irreducible component: \newline $v = v_{m_1} \otimes ... \otimes v_{m_n}$ and consider 
$$ v(\omega_1, ..., \omega_r) := f(\omega_1)...f(\omega_r)v $$

There is a well-known theorem (see \cite{Muk}):
\begin{theorem}{(Mukhin, Tarasov, Varchenko)}

$v(\omega_1, ..., \omega_r)$ is an eigenvector of highest weight $m_1+...+m_n -2r$ if the following Bethe ansatz equations are satisfied: 
$$ \sum^n_{i=1} \frac{m_i}{\omega_l - z_i}  -  \sum^{r}_{s \neq l} \frac{2}{\omega_l - \omega_s} = 0 $$ 
for  each $l = 1,...,r$. 

Moreover, each eigenvector in the space  $(V_{\underline{\lambda}})^{sing}_{m_1+...+m_n-2r}$ can be obtained this way.
\end{theorem}

However, solving this system of $r$ algebraic equations of degree $n+r-2$ is not easy even for small $n$, so alternative methods for obtaining eigenvectors in $(V_{\underline{\lambda}})^{sing}_{m_1+...+m_n-2r}$ should be explored. Our subsequent work can be thought of as simplifying the Bethe ansatz system to only one algebraic equation of degree $\leq r+1$ in the case $n=3$.

\section{Coverings of the Riemann sphere in the three-point Gaudin model}  

Recall that we confined ourselves to the case $\mathfrak{g} = \mathfrak{sl}_2$. Let us also set $n=3$ from now on, so $\underline{\lambda} = (m_1,m_2,m_3)$, where each $m_i$ is a non-negative integer. By observation $\ref{obs}$ our branched coverings $C_{V_{m_1} \otimes V_{m_2} \otimes V_{m_3}} \to \overline{M_{0, 4}}(\mathbb{C})$ splits into disjoint union of coverings $C_{(V_{\underline{\lambda}})^{sing}_{\mu}} \to \overline{M_{0, 4}}(\mathbb{C})$ for each $\mu \in \mathbb{Z}_{\geq 0}$. We will denote such coverings by $C_{(V_{\underline{\lambda}})^{sing}_{\mu}}$.

The space $\overline{M_{0, 4}}(\mathbb{C})$ is nothing but a complex projective line $\mathbb{C}\mathbb{P}^1$ and by making some projective transformation we can map three of our four points-parameters into $0$, $1$, $ \infty $ leaving only one varying parameter $u$.

Let us fix notation for the following elements in $U(\mathfrak{sl_2}) \otimes U(\mathfrak{sl_2}) \otimes U(\mathfrak{sl_2})$: 

\begin{multline*}
\Omega_{12} = e \otimes f \otimes 1 + f \otimes e \otimes 1 + \frac{1}{2} h \otimes h \otimes 1, \quad  \Omega_{23} = 1 \otimes e \otimes f + 1 \otimes f \otimes e + \frac{1}{2} 1 \otimes h \otimes h \\ \Omega_{13} = e \otimes 1 \otimes f + f \otimes 1 \otimes e + \frac{1}{2} h \otimes 1 \otimes h
\end{multline*}

Let us fix $(0,1,\frac{1}{u}, \infty) \in \overline{M_{0, 4}}(\mathbb{C})$. Then the corresponding Gaudin Hamiltonians are listed below:
$$
 H_1 = - \Omega_{12}  -  u \Omega_{13}, \quad H_2 = \Omega_{12}  + \frac{ u}{u-1} \Omega_{23}, \quad  H_3 = u \Omega_{13}   +  \frac{u}{1-u} \Omega_{23}
$$

A natural question arises about the degrees of coverings $C_{(V_{\underline{\lambda}})^{sing}_{\mu}} \to \overline{M_{0, 4}}(\mathbb{C})$. To answer it we need to know dimensions of $(V_{\underline{\lambda}})^{sing}_{\mu}$. They are calculated in the section below.

\section{Degrees of coverings and admissible weights \label{degrees}}

In this section, we consider the 3-fold tensor product $V_{\underline{\lambda}}=V_{m_1}\otimes V_{m_2}\otimes V_{m_3}$ of irreducible $\mathfrak{sl}_2$ representations. The goal of this section is to determine the dimensions of subspaces of singular vectors of a given weight, which in geometric terms are degrees of Gaudin coverings. To achieve this, we decompose $V_{\underline{\lambda}}$ into sum of irreducible representations by first decomposing $V_{m_1}\otimes V_{m_2}$ as $ \oplus^{min(m_1,m_2)}_{r_1 = 0} V_{m_1+m_2-2r_1}$ and then tensor multiplying each term with $V_{m_3}$.

\begin{proposition}

Let $r_1$ be a non-negative integer such that $V_{\underline{\lambda}}$ has a singular vector of weight $m_1+m_2+m_3-2r$ coming from $V_{m_1+m_2-2r_1} \otimes V_{m_3} \subset (V_{m_1} \otimes V_{m_2}) \otimes V_{m_3}$ in the decomposition of the tensor product into direct sum of irreducible representations. The range of all non-negative integers $r_1$ is given by the following inequalities:

\begin{equation}\label{r_1}
\max(0,r-m_3)\leq r_1 \leq \min(r,\min(m_1,m_2), m_1+m_2 -r)
\end{equation}

\end{proposition}

\begin{proof}

By the Clebsch–Gordan formula, a singular vector of weight $m_1+m_2-2r_1$ exists in $V_{m_1} \otimes V_{m_2}$ if and only if $0 \leq r_1 \leq \min(m_1,m_2)$. Applying this to the case $(V_{m_1} \otimes V_{m_2}) \otimes V_{m_3}  \supset V_{m_1+m_2-2r_1} \otimes V_{m_3}$, we get the following inequalities:

\begin{equation*}
\begin{cases}

 0 \leq r_1 \leq \min(m_1,m_2)           \\
0 \leq r- r_1 \leq \min(m_1+m_2-2r_1, m_3)  
\end{cases}
\end{equation*}

These inequalities are equivalent to the system (\ref{r_1}) and can be written as:

\begin{equation} \label{sys_ineq}
\begin{cases}
\begin{aligned}

 0 \leq \;&r_1 \\
 r-m_3 \leq  \; &r_1 \\
 
 &r_1 \leq min(m_1,m_2)       \\
 &r_1\leq m_1+m_2-r \\
 &r_1 \leq r  \\

\end{aligned}
\end{cases}
\end{equation}

 To visualize this formula, consider a plane with coordinates $r$ and $r_1$. For each integer $r$, possible integers $r_1$ are lying on or above the lines $r_1=0$, $r_1=r-m_3$ and on or under the lines $r_1=r$, $r_1= \min(m_1,m_2)$, $r_1 = m_1+m_2-r$. The dimension of the subspace of singular vectors of weight $m_1+m_2+m_3-2r$ is then the number of such integers $r_1$. This is illustrated in Figures 1-4.
\end{proof}

\begin{remark}
    Compare the results of this section with the formula in theorem 1 from \cite{Scher}, which gives negative dimensions even in the non-trivial cases. 
\end{remark}

\begin{center}
\begin{minipage}{0.485\linewidth}
\includegraphics[width=\linewidth]{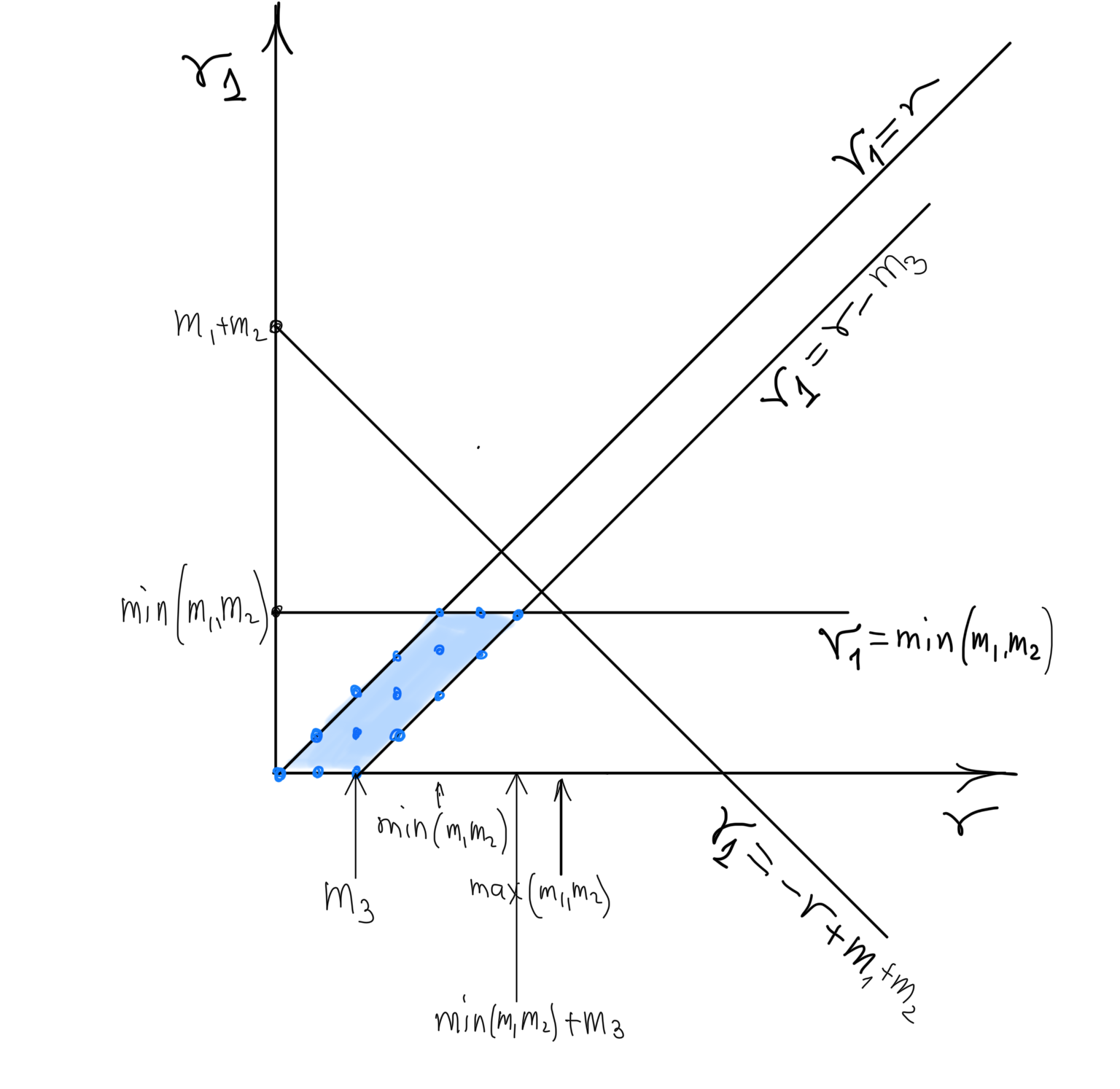}
\captionof{figure}{$\{m_1, m_2\} = \{4,7\}, \: m_3 = 2$}
\end{minipage}%
\hfill
\begin{minipage}{0.49\linewidth}
\includegraphics[width=\linewidth]{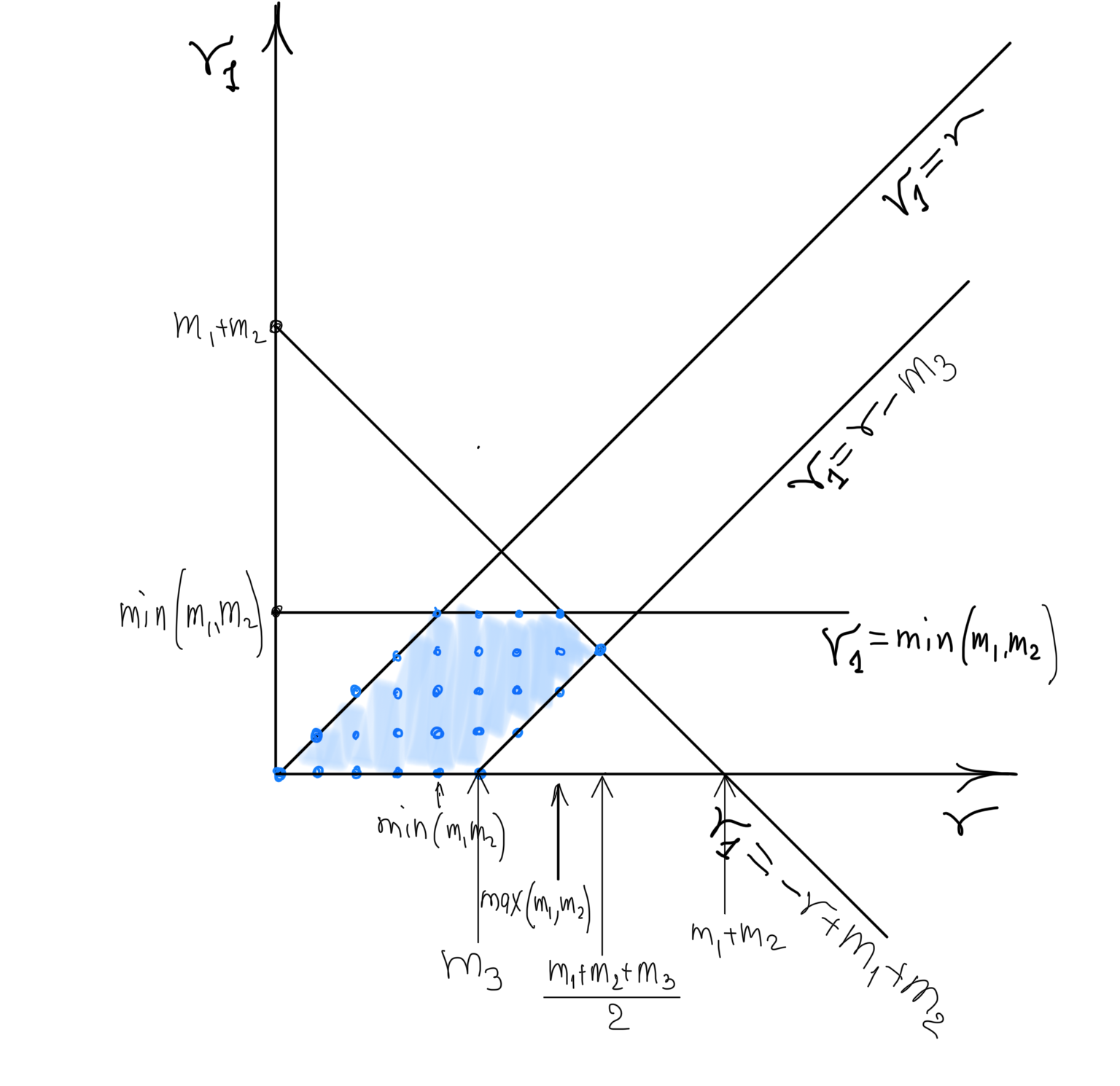}
\captionof{figure}{$\{m_1, m_2\} = \{4,7\}, \: m_3 = 5$}
\end{minipage}
\end{center}

\begin{center}
\begin{minipage}{0.40\linewidth}
\includegraphics[width=\linewidth]{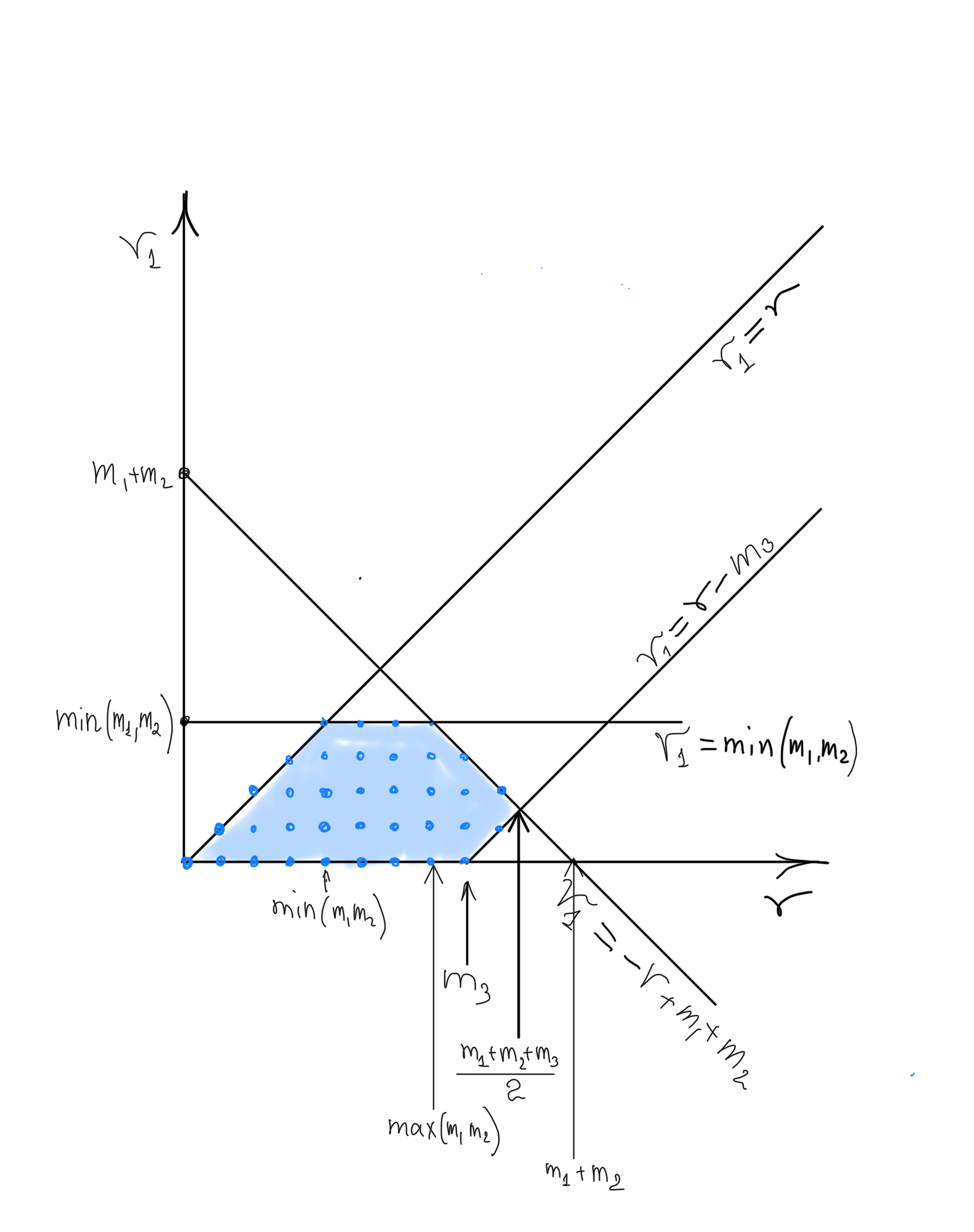}
\captionof{figure}{$\{m_1, m_2\} = \{4,7\}, \: m_3 = 8$}
\end{minipage}
\hfill
\begin{minipage}{0.5\linewidth}
\includegraphics[width=\linewidth]{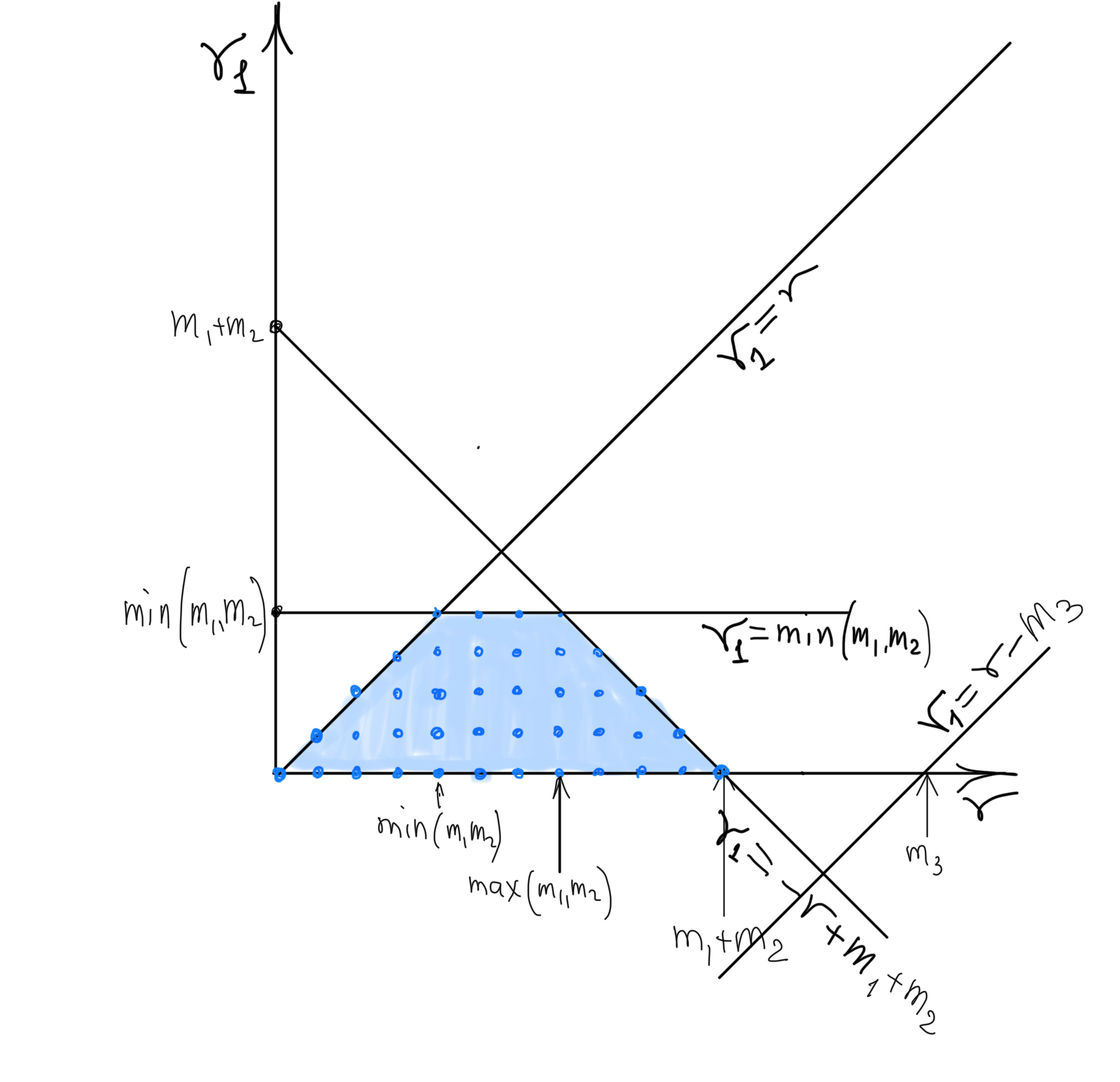}
\captionof{figure}{$\{m_1, m_2\} = \{4,7\}, \: m_3 = 17$}
\end{minipage}
\end{center}

\section{The Shapovalov form \label{Shapovalov}}  

For any non-negative integer $m$ let $V_m$ be the irreducible  $\mathfrak{sl_2}$-module of the highest weight $m$. Let $v_m \in V_m $ be a highest vector of weight $m$, i.e. a vector subject to the conditions $ ev_m = 0, \; hv_m = m v_m$.

\begin{definition} 
For a fixed $v_m$ the Shapovalov form on $V_m$ is the unique symmetric bilinear form $S_m$ such that  
$$ S_m ( v_m, v_m) = 1$$ 
and for any $x,y \in V_{m}$:
$$  S_m( ex, y) = S_m (x, fy) $$

This form is non-degenerate.
\end{definition}

\begin{lemma}\label{shap} The vectors $ f^l v_{m} , \; 0\leq l \leq m  $ are orthogonal with respect to the Shapovalov form.  Moreover, 
$$S(f^l v_{m}, f^p v_{m} )  = \delta_{l,p} \, l! \, (m)_l $$
\end{lemma}
\begin{proof}
Let $l \leq p$ be two non-negatuve integers. Using $e f^{l} v_{m} = l (m - l +1) f^{l-1} v_{m} $ we get: 
\begin{multline*}
S(f^l v_{m},  f^p v_{m}) = S(ef^l v_{m}, f^{p-1} v_{m} ) = \\
=l (m - l +1)   S(f^{l -1} v_{m}, f^{p-1} v_{m} ) = ... =  l! \, (m)_l \, S(v_{m}, f^{p-l} v_{m})
\end{multline*}

Since $ S(v_{m}, v_{m}) = 1$, this is equal to $  l! \, (m)_l $, if $l=p$. If $l<p$ we can increase weights one more time to get that  $ S(f^l v_{m},  f^p v_{m})  = l! \, (m)_l \, S(e v_m, f^{p-l-1} v_m) = l! \, (m)_l \, S(0, f^{p-l-1} v_m) = 0 $.

\end{proof}

\begin{definition} 

For any non-negative integers $m_1, ..., m_n$ the Shapovalov form on the $n$-fold tensor product of irreducible  $\mathfrak{sl_2}$-representations $V_{m_1} \otimes ... \otimes V_{m_n}$ is defined to be $S_{m_1} \otimes ... \otimes S_{m_n} $. We denote it simply by $S$ whenever $m_1,..., m_n$ are clear from the context. Since $S_{m_1}, ...,  S_{m_n}$ are non-degenerate, this form is also non-degenerate. 

\end{definition}

It is easy to see that the Gaudin Hamiltonians $H_i$, which act on $V_{m_1} \otimes ... \otimes  V_{m_n} $ 
are self-adjoint operators with respect to the Shapovalov form:

$$ S(H_ix, y) = S(x, H_iy) $$

This explains the importance of this form in the study of the Gaudin model. In particular, in a basis, that is orthonormal with respect to the Shapovalov form, the matrices of the Gaudin Hamiltonians are symmetric. In what follows, we consider such bases, which we will call the Shapovalov bases.

\begin{lemma}  Shapovalov form is consistent with the diagonal action of $\mathfrak{sl_2}$:
 $$ S(e\cdot v, w) = S(v, f \cdot w), \; \text{where} \; v,w \in V_{m_1} \otimes ... \otimes V_{m_n} $$
\end{lemma}
\begin{proof} 
It is enough to check this equality on the decomposable tensors: 
 $ S(e \cdot (v_1 \otimes ... \otimes v_n) , w_1 \otimes ... \otimes w_n) = S( e v_1 \otimes v_2 \otimes ... \otimes v_n, w_1 \otimes ... \otimes w_n ) + S( v_1 \otimes e v_2 \otimes ... \otimes v_n, w_1 \otimes ... \otimes w_n  ) + ... + S( v_1 \otimes  v_2 \otimes ... \otimes  e v_n ,  w_1 \otimes ... \otimes w_n )  =  S( e v_1 , w_1) S( v_2 , w_2) ... S(v_n, w_n ) + ... + S( v_1, w_1) S(v_2, w_2) ... S(e v_n, w_n) =  S(  v_1 , f w_1) S( v_2 , w_2) ... S(v_n, w_n ) + ... + S( v_1, w_1) S(v_2, w_2) ... S( v_n, f w_n) = S( v_1 \otimes v_2 \otimes ... \otimes v_n, f w_1 \otimes ... \otimes w_n ) + S( v_1 \otimes  v_2 \otimes ... \otimes v_n, w_1 \otimes f w_2 \otimes  ... \otimes w_n  ) + ... + S( v_1 \otimes  v_2 \otimes ... \otimes  v_n ,  w_1 \otimes ... \otimes  f w_n )  =  S(v_1 \otimes ... \otimes v_n  , f \cdot (w_1 \otimes ... \otimes w_n)) .$
\end{proof}

\begin{lemma}\label{orth} Let $V_i, V_j $ be irreducible subrepresentations of $V_{m_1} \otimes ... \otimes V_{m_n}$  of highest weights $i$ and $j$ respectively, $i \neq j$. Then $$ S(V_i , V_j ) = 0  $$
\end{lemma}
\begin{proof}
It is enough to check this equality on the basis vectors $ \{ f^l v_i , 0 \leq l \leq i \} $ and $ \{ f^p v_j , 0 \leq p \leq j \}$ of $V_i$ and $V_j$ respectively. If $l \neq p$ (wlog $l<p$) then we gradually increase weights as in lemma \ref{shap}: 
\begin{multline*}
S(f^l v_{i},  f^p v_{j}) = S(ef^l v_{i}, f^{p-1} v_{j} ) =  l (i - l +1)   S(f^{l -1} v_{i}, f^{p-1} v_{j} ) = \\  \frac{l! \: i!}{(i-l)!}S(v_{i}, f^{p-l} v_{j}) =  \frac{l! \: i!}{(i-l)!} S (e v_i, f^{p-l-1} v_j) = \frac{l! \: i!}{(i-l)!} S (0, f^{p-l-1} v_j) = 0 
\end{multline*}

If  $l = p$ we decrease weights in the same way (wlog assuming $i < j$ ):
\begin{multline*}
S(f^l v_{i},  f^l v_{j}) = \frac{1}{(l+1)(j-l)} S( f^{l} v_{i}, e f^{l+1} v_{j} ) =  \frac{1}{(l+1)(j-l)}  S(f^{l + 1} v_{i}, f^{l+1} v_{j} ) = \\  \frac{1}{(i)_{i-l} (j-l)_{i-l}} S(f^i v_{i}, f^{i} v_{j}) =  \frac{1}{(i)_{i-l+1} (j-l)_{i-l+1}} S (f^i v_i, e f^{i+1} v_j) = \\ \frac{1}{(i)_{i-l+1} (j-l)_{i-l+1}}  S (f^{i+1} v_i, f^{i+1} v_j) = \frac{1}{(i)_{i-l+1} (j-l)_{i-l+1}}  S (0 , f^{i+1} v_j) =  0 
\end{multline*} 

\end{proof} 

Let $V_{m_1}, V_{m_2}$ be the irreducible representations of  $\mathfrak{sl_2}$ of highest weights $m_1$ and $m_2$ respectively, let us also choose highest vectors $v_{m_1}$ and $v_{m_2}$ of weights $m_1$ and $m_2$ in these spaces $V_{m_1}$ and $V_{m_2}$. Compare the following formula with the similar one (though, with a slight misprint in the numerator) in \cite{Var}.

 \begin{proposition} Fixing non-negative integer $l$ such that $l \leq min(m_1,m_2) $ we have a singular vector in $V_{m_1} \otimes V_{m_2} $ of weight $m_1 + m_2 -2l$ given by the following formula:

\begin{equation} \label{eq:first_formula}
\langle v_{m_1}, v_{m_2} \rangle_l :=  \sum_{p=0}^{l}  \frac{ (-1)^p  }  {  (m_1)_p (m_2)_{l-p} }  \frac{f^p v_{m_1}}{p! } \otimes \frac{f^{l-p} v_{m_2}}{ (l-p)!}
\end{equation}

where $(x)_n := x(x-1)...(x-n+1)$ is the Pochhammer symbol.

\end{proposition}

\begin{proof}
    Clearly, the vector $\langle v_{m_1}, v_{m_2} \rangle_l $ has weight $(m_1 -2p) + (m_2-2(l-p)) = m_1+m_2 -2l $. To prove that it is singular we should check that it vanishes as $e$  acts diagonally on it:
\begin{multline*} e \cdot \langle v_{m_1}, v_{m_2} \rangle_l  = \sum_{p=1}^{l} \frac{ (-1)^p  }  {  (m_1)_p (m_2)_{l-p} }  \frac{e f^p v_{m_1}}{p! } \otimes \frac{f^{l-p} v_{m_2}}{ (l-p)!} + \sum_{p=0}^{l-1}  \frac{ (-1)^p  }  {  (m_1)_p (m_2)_{l-p} }  \frac{ f^p v_{m_1}}{p! } \otimes \frac{e f^{l-p} v_{m_2}}{ (l-p)!}  = \\
\sum_{p=1}^{l} \frac{ (-1)^p  }  {  (m_1)_{p-1} (m_2)_{l-p} }  \frac{f^{p-1} v_{m_1}}{(p-1)! } \otimes \frac{f^{l-p} v_{m_2}}{ (l-p)!} + \sum_{p=0}^{l-1}  \frac{ (-1)^p  }  {  (m_1)_p (m_2)_{l-p-1} }  \frac{f^p v_{m_1}}{p! } \otimes \frac{ f^{l-p-1} v_{m_2}}{ (l-p-1)!} = \\
\sum_{p=1}^{l} \frac{ (-1)^p  }  {  (m_1)_{p-1} (m_2)_{l-p} }  \frac{f^{p-1} v_{m_1}}{(p-1)! } \otimes \frac{f^{l-p} v_{m_2}}{ (l-p)!} + \sum_{p=1}^{l}  \frac{ (-1)^{p-1}  }  {  (m_1)_{p-1} (m_2)_{l-p} }  \frac{f^{p-1} v_{m_1}}{(p-1)! } \otimes \frac{ f^{l-p} v_{m_2}}{ (l-p)!} = 0
\end{multline*}
\end{proof}

\begin{lemma} The vectors $\langle v_{m_1}, v_{m_2} \rangle_l, \; 0\leq l \leq min(m_1,m_2)  $ are orthogonal with respect to the Shapovalov form.
\end{lemma}

\begin{proof}
For different $l$ vectors $\langle v_{m_1}, v_{m_2} \rangle_l$ are lying in subrepresentations of different highest weights, so the result follows from lemma \ref{orth}.

\end{proof}

\begin{proposition}
$\norm{ \langle v_{m_1}, v_{m_2} \rangle _{l}}^2_S = \frac{\binom{m_1+m_2-l+1}{l}} {(m_1)_l (m_2)_l} $
\end{proposition}

\begin{proof}
    Since vectors of different weights in an irreducible subrepresentation are orthogonal we have 
\begin{multline*}
 \norm{ \langle v_{m_1}, v_{m_2} \rangle _{l}}^2_S =  \norm{\sum_{p=0}^{l}  \frac{ (-1)^p  }  {  (m_1)_p (m_2)_{l-p} }  \frac{f^p v_{m_1}}{p! } \otimes \frac{f^{l-p} v_{m_2}}{ (l-p)!}}^2_S =\\  \sum_{p=0}^{l}  \frac{\norm{f^p v_{m_1}}^2_S \norm{f^{l-p}{m_2}}^2_S}{((m_1)_p p! (m_2)_{l-p} (l-p)!)^2 } =  \\ 
 \sum_{p=0}^{l}  \frac{1}{(m_1)_p p! (m_2)_{l-p} (l-p)! } = \frac{1}{(m_1)_l (m_2)_l } \sum_{p=0}^{l} \frac{(m_1 - p)_{l-p} }{(l-p)!} \frac{(m_2-l+p)_p}{p!} = \\
 \frac{1 }{(m_1)_l (m_2)_l } \sum_{p=0}^{l} \binom{m_1 - p }{l-p} \binom{m_2-l+p}{p} = \\ \frac{(-1)^l}{(m_1)_l (m_2)_l } \sum_{p=0}^{l} \binom{ - m_1 + l - 1  }{l-p} \binom{ -m_2+ l - 1}{p}  = 
   (-1)^l \frac{\binom{- m_1 - m_2 + 2l- 2}{l}} {(m_1)_l (m_2)_l }  = \frac{\binom{m_1  + m_2  -l + 1 }{l}} {(m_1)_l (m_2)_l } 
\end{multline*}
where we used $ \binom{m+p-1}{p} = (-1)^p \binom{-m}{p} $ and Chu-Vandermonde identity.
\end{proof}

\begin{corollary}\label{norm}
   \begin{equation}
        \norm{ \langle \langle v_{m_1}, v_{m_2} \rangle_{r_1}, v_{m_3} \rangle_{r-r_1}}^2_S = \frac{\binom{m_1+m_2-r_1+1}{r_1}} {(m_1)_{r_1} (m_2)_{r_1}} \, \frac{\binom{m_1+m_2+m_3-r-r_1+1}{r-r_1}} {(m_1+m_2-2r_1)_{r-r_1} (m_3)_{r-r_1}}  
    \end{equation}
\end{corollary}

\begin{proposition}\label{prop_basis}

Let fix non-negative integers $m_1,m_2,m_3,r$. For non-negative integers $r_1, i,j$ we will use the following notation: 

\begin{equation}\label{notation}
C^{\,r_1}_{i,j} := \frac{1}{(m_1)_i (m_2)_{r_1-i} \, i! (r_1-i)! } \frac{1}{(m_1+m_2-2r_1)_j (m_3)_{r-r_1-j} \, j! (r-r_1-j)! } 
\end{equation}

\begin{equation}\label{notation_norm}
N(r_1) := \norm{ \langle \langle v_{m_1}, v_{m_2} \rangle_{r_1}, v_{m_3} \rangle_{r-r_1})}_S
\end{equation}

For any non-negative integers $m_1, m_2, m_3, r$ and admissible $r_1$ the vector

\begin{equation}\label{eq:basis}
 w_{r_1}  = \frac{1}{N(r_1)} \sum_{i=0}^{r_1}
 \sum_{j=0}^{r - r_1} (-1)^{i+j} \, C^{\,r_1}_{i,j} \; f^j (f^{i} v_{m_1} \otimes f^{r_1-i} v_{m_2}) \otimes f^{r-r_1-j} v_{m_3}
\end{equation}

is a singular vector of weight  $m_1 + m_2 + m_3 - 2r$ coming from the summand $V_{m_1 +m_2-2r_1} \otimes V_{m_3} $ in the decomposition of $(V_{m_1} \otimes V_{m_2}) \otimes V_{m_3}$ into the sum of irreducible $\mathfrak{sl_2}$-representations. 

 As before, $(x)_n := x(x-1)...(x-n+1)$ is the Pochhammer symbol.

 Moreover, vectors $\{w_i,  \max(0,r-m_3) \leq i \leq  \min(r,\min(m_1,m_2), m_1+m_2 -r)$ form an orhonormal basis of the space $(V_{m_1} \otimes V_{m_2} \otimes V_{m_3})^{\text{sing}}_{m_1+m_2+m_3-2r}$, which we call (orthonormal) Shapovalov basis. 

\end{proposition}

\begin{proof}[Proof]
The formula \eqref{eq:basis} is just an application of \eqref{eq:first_formula} twice. 

Let us prove orthogonality of vectors $w_{r_1}$. Shapovalov form on tensor product is defined as multiplication of Shapovalov forms on each of the tensor factors. So, since vectors $w_{r_1}$  are of the form $ \sum_i v_i\otimes w_i$ where $v_i \in V_{m_1+m_2-2r_1}, \; w_i \in  V_{m_3}$, we see that for different $r_1$ the first factors belong to subrepresentations of different highest weights, thus these vectors orthogonal by lemma \ref{orth}.

\end{proof}

\section{The action of the Gaudin Hamiltonians on the \newline  Shapovalov bases \label{Shbases} } 
\begin{proposition}\label{dact}
For every non-negative integer $l$ such that $l \leq min(m_1,m_2) $ singular vector $\langle v_{m_1}, v_{m_2} \rangle_l  \in V_{m_1} \otimes V_{m_2}$ is an eigenvector of  $\Omega_{12}$: 
$$\Omega_{12} \langle v_{m_1}, v_{m_2} \rangle_l = \left( l(l-m_1-m_2-1) +\frac{m_1m_2}{2} \right)  \langle v_{m_1}, v_{m_2} \rangle_l $$
\end{proposition}
\begin{proof}
\begin{multline}
\Omega_{12} \langle v_{m_1}, v_{m_2} \rangle_l = \sum_{p=0}^{l}  \frac{ (-1)^p  }   {  (m_1)_p (m_2)_{l-p} } \Omega_{12} (\frac{f^p v_{m_1}}{p! } \otimes \frac{f^{l-p} v_{m_2}}{ (l-p)!}) \\ = \sum_{p=1}^{l} \frac{ (-1)^p  }  {  (m_1)_{p-1} (m_2)_{l-p} } \frac{f^{p-1} v_{m_1}}{(p-1)! } \otimes \frac{f^{l-p+1} v_{m_2}}{ (l-p)!} + \\ \sum_{p=0}^{l-1} \frac{ (-1)^p  }  {  (m_1)_{p} (m_2)_{l-p-1} } \frac{f^{p+1} v_{m_1}}{p! } \otimes \frac{f^{l-p-1} v_{m_2}}{ (l-p-1)!} + \frac{(m_1-2p)(m_2-2l+2p)}{2} \langle v_{m_1}, v_{m_2} \rangle_l = \\
\sum_{p_1 = 0 }^{l-1} \frac{ (-1)^{p_1+1}  }  {  (m_1)_{p_1} (m_2)_{l-p_1-1} } \frac{f^{p_1} v_{m_1}}{p_1! } \otimes \frac{f^{l-p_1} v_{m_2}}{ (l-p_1-1)!} + \sum_{p_2=1}^{l} \frac{ (-1)^{p_2-1}  }  {  (m_1)_{p_2-1} (m_2)_{l-p_2} } \frac{f^{p_2} v_{m_1}}{(p_2-1)! } \otimes \frac{f^{l-p_2} v_{m_2}}{ (l-p_2)!} + \\  \sum_{p=0}^{l}  \frac{ (-1)^p  }   {  (m_1)_p (m_2)_{l-p} } \frac{(m_1-2p)(m_2-2l+2p)}{2} \frac{f^p v_{m_1}}{p! } \otimes \frac{f^{l-p} v_{m_2}}{ (l-p)!}   = \\
 \sum_{p=0}^{l}  \frac{ (-1)^p  }{(m_1)_p (m_2)_{l-p}} (-(l-p)(m_2-l+p+1) - p(m_1-p+1) + \\ \frac{(m_1-2p)(m_2-2l+2p)}{2})\frac{f^p v_{m_1}}{p! } \otimes \frac{f^{l-p} v_{m_2}}{ (l-p)!} = (l(l-m_1-m_2-1) +\frac{m_1m_2}{2}) \langle v_{m_1}, v_{m_2} \rangle_l
\end{multline}
\end{proof}

\begin{corollary}\label{dact_1}
 In the orthonormal Shapovalov basis the matrix of the operator $\Omega_{12}$ restricted to the invariant subspace of singular vectors of weight $m_1+m_2+m_3 -2r$ is the following diagonal matrix.
\begin{equation}\label{dactom}
\Omega_{12} = diag[d_{r_1}]_{r_1} := diag[r_1(r_1-m_1-m_2-1)]_{r_1} + \frac{m_1m_2}{2}
\end{equation} 
where $\max(0,r-m_3)\leq r_1 \leq \min(r,\min(m_1,m_2), m_1+m_2 -r)$.
\end{corollary}

\begin{proposition}

$\Omega_{23}$ and $\Omega_{13}$ act as symmetric tridiagonal matrices in the Shapovalov basis. 
    
\end{proposition}

\begin{proof}

From the formula \eqref{eq:basis} we see that vector $w_{r_1}$ is a linear combination of standard basis vectors of the form $f^{i} v_{m_1} \otimes f^{j}  v_{m_2} \otimes f^{k} v_{m_3}$, where $k = r_1, \, r_1+1, ..., \min(r,\min(m_1,m_2), m_1+m_2 -r)$. Thus, $\Omega_{23} (w_{r_1})$ is a linear combination of  $ w_{r_1-1}, w_{r_1}, ...,  w_{\min(r,\min(m_1,m_2), m_1+m_2 -r)} $, since the third tensor component can get only one degree lower in $f$ by applying $\Omega_{23}$ or $\Omega_{13}$. In the orthonormal Shapovalov basis $\Omega_{23}$ and $\Omega_{13}$ are given by symmetric matrices since these operators are self-adjoint with respect to the Shapovalov form. It follows that the both vectors $\Omega_{23}(w_{r_1})$ $\Omega_{13}(w_{r_1})$ are linear combinations of vectors $ w_{r_1-1}, w_{r_1}, w_{r_1+1} $. 
\end{proof}

\begin{corollary}
    Gaudin Hamiltonians $H_1, H_2, H_3$ are symmetric tridiagonal matrices in the Shapovalov basis. 
\end{corollary}

\begin{theorem}
In the orthonormal Shapovalov basis matrices of the operators $\Omega_{2,3},\Omega_{1,3}$ restricted to the invariant subspace of singular vectors of weight $m_1+m_2+m_3 -2r$ are the following symmetric irreducible tridiagonal matrices.

\begin{equation}\label{eq:matrix}
\Omega_{23} 
= \begin{pmatrix}
 \ddots &  \ddots \\
 \ddots &  b_{r_1} &    c_{r_1} \\
&  c_{r_1} & \ddots  &  \ddots  \\
& & \ddots   & \ddots   
\end{pmatrix} \quad
\Omega_{13} 
= \begin{pmatrix}
 \ddots &  \ddots \\
 \ddots &  a_{r_1} &  - c_{r_1} \\
&  - c_{r_1} & \ddots  &  \ddots  \\
& & \ddots   & \ddots   
\end{pmatrix} 
\end{equation}

where $\max(0,r-m_3)\leq r_1 \leq \min(r,\min(m_1,m_2), m_1+m_2 -r)$, and:

\begin{equation}\label{eq:mid13}
\begin{gathered}
a_{r_1} = \frac{m_1(m_3-2r)}{2} + r_1(m_1-m_3+2r-2r_1) + \\
 \frac{r_1(m_1-r_1+1)(r-r_1+1)(m_3-r+r_1)}{m_1+m_2-2r_1+2}
 - \frac{(r_1+1)(m_1-r_1)(r-r_1)(m_3-r+r_1+1)}{m_1+m_2-2r_1} 
\end{gathered}
\end{equation}

\begin{equation}\label{eq:mid23}
\begin{gathered}
b_{r_1} = \frac{m_2(m_3-2r)}{2} + r_1(m_2-m_3+2r-2r_1) + \\
 \frac{r_1(m_2-r_1+1)(r-r_1+1)(m_3-r+r_1)}{m_1+m_2-2r_1+2}
 - \frac{(r_1+1)(m_2-r_1)(r-r_1)(m_3-r+r_1+1)}{m_1+m_2-2r_1} 
\end{gathered}
\end{equation}

\begin{equation}\label{eq:lateralom}
\begin{gathered}
c_{r_1}^2 = \frac{(r_1+1)(r-r_1)(m_1-r_1)(m_2-r_1)(m_1+m_2-r_1+1)(m_3-r+r_1+1)}{(  m_1+m_2-2r_1 -1)(m_1+m_2-2r_1)^2} \times \\ \times \frac{(m_1+m_2-r-r_1)(m_1+m_2+m_3-r-r_1+1)}{(m_1+m_2-2r_1+1)}
\end{gathered}
\end{equation}

\end{theorem}

\begin{remark} 
It is easy to see that $c_{r_1}^2$ is always positive and none of the denominators in the formulas (\ref{eq:mid13}), (\ref{eq:mid23}), (\ref{eq:lateralom}) vanish under the existing restrictions (\ref{r_1}) on $r_1$. 
\end{remark}

\begin{proof} 
Starting from $\Omega_{23}$ let us now find the exact entries of these matrices.

To find  the coefficient  $b_{r_1, \: r_1 -1} $ of the vector $\Omega_{23} (w_{r_1})$ in front of $w_{r_1 -1}$ we examine the component of the vector   $ v_{m_1} \otimes f^{r_1-1} v_{m_2} \otimes f^{r-r_1+1} v_{m_3}$ in the decomposition of $\Omega_{23}(w_{r_1})$ in the standard basis and we compare it with the component of the same vector in the decomposition $b_{r_1, \: r_1 -1} \: w_{r_1-1}$ in the same basis:

On the one hand, when $\Omega_{23}$ acts on $w_{r_1}$ only  $  v_{m_1} \otimes f^{r_1} v_{m_2} \otimes f^{r-r_1} v_{m_3}$ contributes to this component, on which $\Omega_{23}$ acts as $ 1\otimes e \otimes f$. 

We find the coefficient in front of  $  v_{m_1} \otimes f^{r_1} v_{m_2} \otimes f^{r-r_1} v_{m_3}$ in the decomposition of $w_{r_1}$ using the formula \eqref{eq:basis}. Then, considering that $e f^{r_1} v_{m_2} = r_1 (m_2 - r_1 +1) f^{r_1-1} v_{m_2} $, we multiply it by  $ r_1 (m_2 - r_1 +1) $ and we result in:

\begin{equation}\label{eq:om23up} r_1 (m_2 - r_1 +1) \: \frac{ C^{\,r_1}_{0,0} }{ N(r_1)}
\end{equation}

From the other hand, the component $  v_{m_1} \otimes f^{r_1-1} v_{m_2} \otimes f^{r-r_1+1} v_{m_3}$ in the decomposition of $b_{r_1, \: r_1 -1} \: w_{r_1-1}$ by \eqref{eq:basis} is equal to: 

\begin{equation}\label{eq:om23up_1}  b_{r_1, \: r_1 -1}  \: \frac{ C^{\,r_1-1}_{0,0} }{N(r_1-1)}
\end{equation}

Thus, comparing \eqref{eq:om23up} and \eqref{eq:om23up_1} what we get is: 

\begin{multline}\label{eq:om23_upper_2}
b_{r_1, \: r_1 -1} \:=   r_1 (m_2 - r_1 +1) \frac{C^{\,r_1}_{0,0}}{C^{\,r_1-1}_{0,0}}  \: \frac{N(r_1-1)}{N(r_1)}  = (r-r_1+1) (m_3-r+r_1)     \: \frac{N(r_1-1)}{N(r_1)}  = \\   \sqrt{\frac{(r-r_1+1) (m_3-r+r_1)  (m_1+m_2-r-r_1+1) (m_1+m_2+m_3-r-r_1+2) }{(m_1+m_2-2r_1+1)(m_1+m_2-2r_1+2) }} \times \\ \times \frac{\norm{  \langle v_{m_1}, v_{m_2} \rangle_{r_1-1}}_S}{\norm{\langle v_{m_1}, v_{m_2} \rangle_{r_1}}_S} = \\
\sqrt{\frac{(r-r_1+1) (m_3-r+r_1)  (m_1+m_2-r-r_1+1) (m_1+m_2+m_3-r-r_1+2)}{(m_1+m_2-2r_1+1)(m_1+m_2-2r_1+2)} } \times \\ \times \sqrt{\frac{r_1(m_1-r_1+1)(m_2-r_1+1)(m_1+m_2-r_1+2) }{(m_1+m_2-2r_1+2) (m_1+m_2-2r_1+3)}}   
\end{multline}

Now, to find the coefficient $b_{r_1, \: r_1}$ of the vector $\Omega_{23}(w_{r_1})$ in front of $w_{r_1}$ we examine the component of the vector $  v_{m_1} \otimes f^{r_1} v_{m_2} \otimes f^{r-r_1} v_{m_3}$ in the decomposition of $\Omega_{23}(w_{r_1})$ in the standard basis and we compare it with the component of the same vector in the decomposition of  $b_{r_1, \: r_1-1} \: w_{r_1 -1} + b_{r_1, \: r_1} \: w_{r_1} $  in the same basis (from formula \eqref{eq:basis} it is clear that  $w_{r_1+1}$ does not contribute to this component).

On the one hand, when $\Omega_{23}$ acts on $w_{r_1}$ the contribution to this component comes from $  v_{m_1} \otimes f^{r_1} v_{m_2} \otimes f^{r-r_1} v_{m_3}$, on which $\Omega_{23}$ acts as $  1 \otimes h \otimes h   $, and from $  v_{m_1} \otimes f^{r_1+1} v_{m_2} \otimes f^{r-r_1-1} v_{m_3}$, on which $\Omega_{23}$ acts as $   1\otimes e \otimes f$.  

Putting this all together using formula \eqref{eq:basis} we get that the component of the vector \newline $ v_{m_1} \otimes f^{r_1} v_{m_2} \otimes f^{r-r_1} v_{m_3}$ in the decomposition of $\Omega_{23}(w_{r_1})$ in the standard basis is equal to:

\begin{equation}\label{eq:23c}
\begin{gathered}
\frac{(m_2-2r_1)(m_3-2r+2r_1)}{2} \frac{\,C^{r_1}_{0,0}}{N(r_1)} - (r_1+1)(m_2-r_1) \frac{\,C^{r_1}_{0,1}}{N(r_1)}
\end{gathered}
\end{equation}

From the other hand, the component $ f^{r_1} v_{m_1} \otimes v_{m_2} \otimes f^{r-r_1} v_{m_3}$ in the decomposition of $b_{r_1, \: r_1-1} \: w_{r_1 -1} + b_{r_1, \: r_1} \: w_{r_1} $ in the standard basis is equal to:

\begin{equation}\label{eq:23d}
\begin{gathered}
 b_{r_1, \: r_1}  \frac{C^{\,r_1}_{0,0}}{N(r_1)}  -  b_{r_1, \: r_1 -1} \frac{C^{\,r_1-1}_{0,1}}{N(r_1-1)}
\end{gathered}
\end{equation}

 Thus, comparing  \eqref{eq:23c} and \eqref{eq:23d}, taking  \eqref{eq:om23_upper_2} into account we get:

\begin{equation}\label{eq:23middle}
\begin{gathered}
b_{r_1, \: r_1}  =  \frac{(m_2-2r_1)(m_3-2r+2r_1)}{2} + b_{r_1, \: r_1 -1} \frac{N(r_1)}{N(r_1-1)} \frac{C^{\,r_1-1}_{0,1}}{C^{\,r_1}_{0,0}}   - (r_1+1)(m_2-r_1) \frac{C^{\,r_1}_{0,1}}{C^{\,r_1}_{0,0}}    = 
\\
\frac{(m_2-2r_1)(m_3-2r+2r_1)}{2} + r_1(m_2-r_1+1) \frac{C^{\,r_1-1}_{0,1}}{C^{\,r_1-1}_{0,0}}  - (r_1+1)(m_2-r_1) \frac{C^{\,r_1}_{0,1}}{C^{\,r_1}_{0,0}} 
=
\\
\frac{m_2(m_3-2r)}{2} + r_1(m_2-m_3+2r-2r_1) +\\
 \frac{r_1(m_2-r_1+1)(r-r_1+1)(m_3-r+r_1)}{m_1+m_2-2r_1+2}
 - \frac{(r_1+1)(m_2-r_1)(r-r_1)(m_3-r+r_1+1)}{m_1+m_2-2r_1} 
\end{gathered}
\end{equation}

Let us now proceed to find exact entries of $\Omega_{13}$. 

Similarly to what we did before, to find  the coefficient  $a_{r_1, \: r_1 -1} $ of the vector $\Omega_{13} (w_{r_1})$ in front of $w_{r_1 -1}$ we examine the component of the vector   $ f^{r_1-1} v_{m_1} \otimes  v_{m_2} \otimes f^{r-r_1+1} v_{m_3}$ in the decomposition of $\Omega_{13}(w_{r_1})$ in the standard basis and we compare it with the component of the same vector in the decomposition $b_{r_1, \: r_1 -1} \: w_{r_1-1}$ in the same basis:

On the one hand, when $\Omega_{13}$ acts on $w_{r_1}$ only  $  f^{r-r_1}  v_{m_1} \otimes f^{r_1} v_{m_2} \otimes v_{m_3}$ contributes to this component, on which $\Omega_{13}$ acts as $ e\otimes 1 \otimes f$. 

We find the coefficient in front of  $  v_{m_1} \otimes f^{r_1} v_{m_2} \otimes f^{r-r_1} v_{m_3}$ in the decomposition of $w_{r_1}$ using the formula \eqref{eq:basis}. Then, considering that $e f^{r_1} v_{m_1} = r_1 (m_1 - r_1 +1) f^{r_1-1} v_{m_1} $, we multiply it by  $ r_1 (m_1 - r_1 +1) $ and we result in:

\begin{equation}\label{eq:a_2} (-1)^{r_1} \: r_1 (m_1 - r_1 +1) \: \frac{ C^{\,r_1}_{r_1,0} }{ N(r_1)}
\end{equation}

From the other hand, the component $  f^{r_1-1} v_{m_1} \otimes  v_{m_2} \otimes f^{r-r_1+1} v_{m_3}$ in the decomposition of $a_{r_1, \: r_1 -1} \: w_{r_1-1}$ by \eqref{eq:basis} is equal to: 

\begin{equation}\label{eq:b_2}  (-1)^{r_1-1} \: b_{r_1, \: r_1 -1}  \: \frac{ C^{\,r_1-1}_{r_1-1,0} }{N(r_1-1)}
\end{equation}

Thus, comparing \eqref{eq:a_2} and \eqref{eq:b_2} what we get is: 

\begin{multline}\label{eq:upper_2}
a_{r_1, \: r_1 -1} \:=    - r_1 (m_1 - r_1 +1) \frac{C^{\,r_1}_{r_1,0}}{C^{\,r_1-1}_{r_1-1,0}}  \: \frac{N(r_1-1)}{N(r_1)}  = (r-r_1+1) (m_3-r+r_1)     \: \frac{N(r_1-1)}{N(r_1)}  = \\   - \sqrt{\frac{(r-r_1+1) (m_3-r+r_1)  (m_1+m_2-r-r_1+1) (m_1+m_2+m_3-r-r_1+2) }{(m_1+m_2-2r_1+1)(m_1+m_2-2r_1+2) }} \times \\ \times \frac{\norm{  \langle v_{m_1}, v_{m_2} \rangle_{r_1-1}}_S}{\norm{\langle v_{m_1}, v_{m_2} \rangle_{r_1}}_S} = \\
- \sqrt{\frac{(r-r_1+1) (m_3-r+r_1)  (m_1+m_2-r-r_1+1) (m_1+m_2+m_3-r-r_1+2)}{(m_1+m_2-2r_1+1)(m_1+m_2-2r_1+2)} } \times \\ \times \sqrt{\frac{r_1(m_1-r_1+1)(m_2-r_1+1)(m_1+m_2-r_1+2) }{(m_1+m_2-2r_1+2) (m_1+m_2-2r_1+3)}}   = - b_{r_1, \: r_1 -1}
\end{multline}

Now, to find the coefficient $a_{r_1, \: r_1}$ of the vector $\Omega_{13}(w_{r_1})$ in front of $w_{r_1}$ we examine the component of the vector $ f^{r_1}  v_{m_1} \otimes  v_{m_2} \otimes f^{r-r_1} v_{m_3}$ in the decomposition of $\Omega_{13}(w_{r_1})$ in the standard basis and we compare it with the component of the same vector in the decomposition of  $a_{r_1, \: r_1-1} \: w_{r_1 -1} + a_{r_1, \: r_1} \: w_{r_1} $  in the same basis (from formula \eqref{eq:basis} it is clear that  $w_{r_1+1}$ does not contribute to this component).

On the one hand, when $\Omega_{13}$ acts on $w_{r_1}$ the contribution to this component comes from $  f^{r_1} v_{m_1} \otimes  v_{m_2} \otimes f^{r-r_1} v_{m_3}$, on which $\Omega_{13}$ acts as $  h \otimes 1 \otimes h   $, and from $  f^{r_1+1} v_{m_1} \otimes  v_{m_2} \otimes f^{r-r_1-1} v_{m_3}$, on which $\Omega_{13}$ acts as $   e\otimes 1 \otimes f$.  

Putting this all together using formula \eqref{eq:basis} we get that the component of the vector \newline $ v_{m_1} \otimes f^{r_1} v_{m_2} \otimes f^{r-r_1} v_{m_3}$ in the decomposition of $\Omega_{13}(w_{r_1})$ in the standard basis is equal to:

\begin{equation}\label{eq:c}
\begin{gathered}
(-1)^{r_1} \: \frac{(m_1-2r_1)(m_3-2r+2r_1)}{2} \frac{\,C^{r_1}_{r_1,0}}{N(r_1)} + (-1)^{r_1+1} \: (r_1+1)(m_1-r_1) \frac{\,C^{r_1}_{r_1,1}}{N(r_1)}
\end{gathered}
\end{equation}

From the other hand, the component $ f^{r_1} v_{m_1} \otimes v_{m_2} \otimes f^{r-r_1} v_{m_3}$ in the decomposition of $a_{r_1, \: r_1-1} \: w_{r_1 -1} + a_{r_1, \: r_1} \: w_{r_1} $ in the standard basis is equal to:

\begin{equation}\label{eq:d}
\begin{gathered}
 (-1)^{r_1} \: a_{r_1, \: r_1}  \frac{C^{\,r_1}_{r_1,0}}{N(r_1)} + (-1)^{r_1} \:  a_{r_1, \: r_1 -1} \frac{C^{\,r_1-1}_{r_1-1,1}}{N(r_1-1)}
\end{gathered}
\end{equation}

 Thus, comparing  \eqref{eq:c} and \eqref{eq:d}, taking  \eqref{eq:upper_2} into account we get:

\begin{equation}\label{eq:middle}
\begin{gathered}
a_{r_1, \: r_1}  =  \frac{(m_1-2r_1)(m_3-2r+2r_1)}{2} - a_{r_1, \: r_1 -1} \frac{N(r_1)}{N(r_1-1)} \frac{C^{\,r_1-1}_{r_1-1,1}}{C^{\,r_1}_{r_1,0}}   - (r_1+1)(m_1-r_1) \frac{C^{\,r_1}_{r_1,1}}{C^{\,r_1}_{r_1,0}}    = 
\\
\frac{(m_1-2r_1)(m_3-2r+2r_1)}{2} + r_1(m_1-r_1+1) \frac{C^{\,r_1-1}_{r_1-1,1}}{C^{\,r_1-1}_{r_1-1,0}}  - (r_1+1)(m_1-r_1) \frac{C^{\,r_1}_{r_1,1}}{C^{\,r_1}_{r_1,0}} 
=
\\
\frac{m_1(m_3-2r)}{2} + r_1(m_1-m_3+2r-2r_1) + \\
 \frac{r_1(m_1-r_1+1)(r-r_1+1)(m_3-r+r_1)}{m_1+m_2-2r_1+2}
 - \frac{(r_1+1)(m_1-r_1)(r-r_1)(m_3-r+r_1+1)}{m_1+m_2-2r_1} 
\end{gathered}
\end{equation}

\end{proof}

\begin{corollary}\label{c_123}
The sum of operators  $\Omega_{12}, \Omega_{23},\Omega_{13}$ act as the following scalar matrix on the space of singular vectors of weight $m_1+m_2+m_3 -2r$.
\begin{multline}
\Omega_{12}+\Omega_{23} +\Omega_{13}  =  \frac{1}{4} (m_1+m_2+m_3-2r)(m_1+m_2+m_3-2r+2) - \\ - \frac{1}{4} \big(m_1(m_1+2) -m_2(m_2+2) -m_3(m_3+2) \big ) =\\
 -r(m_1+m_2+m_3 -r+1) + \frac{1}{4} ((m_1+m_2+m_3)(m_1+m_2+m_3+2) -\\- m_1(m_1+2) -m_2(m_2+2) -m_3(m_3+2))
\end{multline}
\end{corollary}

\begin{corollary}
In the orthonormal Shapovalov basis the matrices of the operators $H_1, H_2, H_3$ restricted to the invariant subspace of singular vectors of weight $m_1+m_2+m_3 -2r$ with a point $(0,1,\frac{1}{u}, \infty) \in \overline{M_{0,4}(\mathbb{C})}$ chosen in the parameter space of the Gaudin model are the following symmetric tridiagonal matrices:

\begin{equation}\label{eq:ham_matrix}
H_1 := \begin{pmatrix}
 \ddots &  \ddots \\
 \ddots &  - d_{r_1} - u \: a_{r_1} &  u \: c_{r_1} \\
& u \: c_{r_1} & \ddots  &  \ddots  \\
& & \ddots   & \ddots   
\end{pmatrix} 
H_2 := \begin{pmatrix}
 \ddots &  \ddots \\
 \ddots & d_{r_1} + \frac{u}{u-1} \: b_{r_1} &  \frac{u}{u-1} \: c_{r_1} \\
& \frac{u}{u-1} \: c_{r_1} & \ddots  &  \ddots  \\
& & \ddots   & \ddots   
\end{pmatrix}
\end{equation} 

where $\max(0,r-m_3)\leq r_1 \leq \min(r,\min(m_1,m_2), m_1+m_2 -r)$, and $d_{r_1}, a_{r_1}, b_{r_1}, c_{r_1}$ are given by formulas (\ref{dactom}), (\ref{eq:mid13}), (\ref{eq:mid23}), (\ref{eq:lateralom}). As previously mentioned, $H_3 = -H_1 - H_2$.

\end{corollary}

\section{Tridiagonal matrices and Gaudin algebraic curves \label{Tridiagonal}}

\begin{definition}

Consider characteristic polynomial of $H_1$, which is $f(x,u) =  det(xId+u\Omega_{13}+\Omega_{12})$. We call the Gaudin curve an algebraic curve given by $f(x,u) = 0$. When it is convenient or necessary we also consider the projectivization of the Gaudin curve: 

$$ f(x,u,w) = det (xId+u\Omega_{13}+w\Omega_{12}) = 0$$

\end{definition}

The tridiagonal form of the Gaudin Hamiltonian enables us to efficiently calculate the Gaudin curves.

\begin{proposition} 
Set $N:=\min(r,\min(m_1,m_2), m_1+m_2 -r)$. Characteristic polynomial $f(x,u,w)$ of the Gaudin Hamiltonian acting on singular vectors of weight $m_1+m_2+m_3-2r$ can be found as $f_N(x,u,w)$ in the following three-term recurrence relation:

\begin{equation}
f_{r_1}(x,u,w) = 
\begin{cases}
    0 ,\; \text{if } r_1 < \max(0,r-m_3)  \\

    1 ,\; \text{if } r_1 = \max(0,r-m_3)  \\
   (x  + u \, a_{r_1} + w \, d_{r_1} ) f_{r_1-1}(x,u,w)   - u^2 \, c_{r_1-1}^2 \, f_{r_1-2}(x,u,w), \, \text{if } r_1 \leq N  
\end{cases}
\end{equation}
\begin{proof}
Standard application of cofactor expansion to the tridiagonal matrix $xId+u\Omega_{13}+w\Omega_{12}$.
\end{proof}
\end{proposition}

We can also construct eigenvectors of the Gaudin Hamiltonian. 

Let $v(x,u,w)$ be the following vector
\begin{align}\label{v}
    v(x,u,w) &= \begin{pmatrix}
           \vdots  \\
           (-1)^{N-r_1} \: c_{r_1} c_{r_1+1} ... c_{N-1} \: f_{r_1-1}(x,u,w)\\
           (-1)^{N-r_1-1} \:  c_{r_1+1} ... c_{N-1} \: f_{r_1}(x,u,w) \\
           \vdots  \\
           f_{N-1}(x,u,w)
         \end{pmatrix},  \;  r_1 \geq max(0,r-m_3) + 1
  \end{align} 

\begin{theorem}
    Vector $v(x,u,w)$ given by (\ref{v}) is an eigenvector of the three-point Gaudin model with eigenvalue $x$ if and only if $(x:u:w)$ lies on the Gaudin curve: $f(x,u,w) = 0$.
\end{theorem}

\begin{proof}
Direct application of vector $v(x,u,w)$ to the matrix $xId+u\Omega_{13}+w\Omega_{12}$ results in a vector in which all coordinates except the last one are zero, and the last one is equal to $f(x,u,w)$.
\end{proof}

Let us also proceed to describe Gaudin coverings.

\begin{proposition}
The Gaudin covering $C_{(V_{\underline{\lambda}})^{sing}_{m_1+m_2+m_3-2r}}$ is just a projection given by $$\{(x:u:w), \: \text{such that} \: f(x,u,w) = 0\} \to \overline{M_{0, 4}}(\mathbb{C}), \; (x:u:w) \mapsto (u:w)$$

\end{proposition}

In the work of Rybnikov \cite{R1} (corollary $3.19$) it was proved that for for any collection $\underline{\lambda}$ of dominant integral weights the spectra of the algebras $\mathcal{A}(\underline{z})$ in the space $V_{\underline{\lambda}}$ form an unbranched covering of $\overline{M_{0, n+1}}(\mathbb{R})$ when parameters $\underline{z}$ are real.

From our results, we can see directly why it is the case when $n=3$. 

Indeed, from the representation of the Gaudin Hamiltonians $H_1$ as tridiagonal matrices in \eqref{eq:ham_matrix} it follows that in the case $u=0$ the matrix is diagonal, and there are pairwise different entries on the diagonal  $$ - d_{r_1} = - r_1(r_1-m_1-m_2-1) - \frac{m_1m_2}{2}$$
Indeed, $d_{r_1} = d_{r_1'}$ implies $(r_1-r_1')(r_1+r_1'-m_1-m_2-1) = 0 $, and $r_1+r_1'-m_1-m_2-1$ is non-zero, since both $r_1$ and $r_1'$ are less than $\min(m_1,m_2)$ from inequality (\ref{r_1}).

In the case  $u \neq 0$ we can delete the first row and first column of the matrix \eqref{eq:ham_matrix}. The resulting matrix is upper-triangular:
\newcommand\x{\times}
\newcommand\bigzero{\makebox(0,0){\text{\huge0}}}
\newcommand*{\bord}{\multicolumn{1}{c|}{}}
$$
  \left(
    \begin{array}{ccccc}
     \ddots    & \ddots      & \ddots   &     &   \\ 
     & \ddots       & -d_{r_1}- u \,a_{r_1}      & u \, c_{r_1}    &  \\ 
          \bigzero&    & u \, c_{r_1}   & \ddots    &  \ddots  \\ 
           &  &  & \ddots   & \ddots    \\ 
 
  \end{array}\right)
$$

with non-zero elements $u \, c_{r_1}$ on the diagonal. It follows that $\text{rank} (x Id - H_1 ) \geq n -1$ for any $x \in \mathbb{C}$, and thus the matrix $H_1$ in the Jordan canonical form has only one Jordan cell for each eigenvalue. When $u$ is real the matrix  $H_1$ is diagonalizable as a real symmetric matrix, and thus  it should have a simple spectrum. 

From this reasoning above we also get the following description of the branch points of the Gaudin coverings: when some eigenvalues collide their eigenspace is one-dimensional.

To determine branch points of the Gaudin covering we compute discriminant. The branch points are given as the zeroes of the discriminant of $f$ with respect to $x$ (recall that there are no real branch points, and this includes the point at infinity):
$$ \text{Resultant} (f(x,u),   \frac{ \partial f(x,u)}{\partial x}) = 0 $$.

The following proposition gives a description of critical points of the Gaudin covering in terms of the vector $v(x,u,w)$.

\begin{proposition}
The point $(x:u:w)$ on a Gaudin curve is a critical point of the Gaudin covering if and only if the vector $v(x,u,w)$ is isotropic, meaning 
$$ v(x,u,w)^T \, v(x,u,w) = 0$$

\begin{proof}
    Note that $v(x,u,w) = adj(xId+u\Omega_{13}+w\Omega_{12})e_N$. Therefore, we have 

    \begin{multline}(xId+u\Omega_{13}+w\Omega_{12}) v(x,u,w) = (xId+u\Omega_{13}+w\Omega_{12}) \: adj(xId+u\Omega_{13}+w\Omega_{12})e_N = \\ det(xId+u\Omega_{13}+w\Omega_{12})e_N 
    \end{multline}

    Differentiating with respect to $x$:
    $$v(x,u,w) + (xId+u\Omega_{13}+w\Omega_{12}) \frac{ \partial v(x,u,w)}{\partial x} = \frac{ \partial \: det(xId+u\Omega_{13}+w\Omega_{12})}{\partial x} e_N $$ 

Multiplying by $v(x,u,w)^T$ from the left and using that $v(x,u,w)^T (xId+u\Omega_{13}+w\Omega_{12}) = 0 $ iff $(x:u:w)$ lies on a Gaudin curve, we get the desired claim.

\end{proof}

\end{proposition}

\section{Hypotheses and observations \label{Hypotheses}}

\subsection{Algebraic curves of the Gaudin coverings are all smooth}

\begin{hypothesis}

For each triple of non-negative integers $ \underline{\lambda}  = (m_1, m_2, m_3)$ and for each admissible non-negative integer weight $\mu$ algebraic curves of the Gaudin covering by the joint spectrum on $(V_{\underline{\lambda}}  ) _{\mu}^{\text{sing}}$ is non-singular. This hypothesis is verified in numerous examples (see section 10).  

\end{hypothesis}

\subsection{All branch points of the Gaudin coverings are simple}

\begin{hypothesis}

For each triple of non-negative integers $ \underline{\lambda}  = (m_1, m_2, m_3)$ and for each admissible non-negative integer weight $\mu$ the Gaudin covering of $(V_{\underline{\lambda}}) _{\mu}^{\text{sing}}$ has only simple branch points. In other words, the system $f(x,u) =   \frac{ \partial f(x,u)}{\partial x} =  \frac{ \partial^2 f(x,u)}{\partial x^2} = 0 $ is never solvable. This hypothesis is verified in numerous examples (see section 10).  

\end{hypothesis}

\subsection{Asymptotic behavior of ornaments of the Gaudin covering}

\begin{proposition}

For a triple of positive integers $ \underline{\lambda}  = (m_1, m_2, m_3)$  algebraic curve of the 2-fold Gaudin covering of $(V_{\underline{\lambda}}) _{m_1+m_2+m_3-2}^{\text{sing}}$ is given by the following equation 
 \begin{multline}\label{eq:2-fold}
x^2 + (((m_1 - 1)m_3 - m_1)u + (m_1 - 1)m_2 - m_1)x + ((\frac{1}{4}m_1^2 - \frac{1}{2}m_1)m_3^2 - \frac{1}{2}m_1^2m_3)u^2 + \\ (((\frac{1}{2}m_1^2 - m_1)m_2 - \frac{1}{2}m_1^2 + m_1)m_3 + (-\frac{1}{2}m_1^2 + m_1)m_2 + m_1^2)u + (\frac{1}{4}m_1^2 - \frac{1}{2}m_1)m_2^2 - \frac{1}{2}m_1^2m_2  = 0
\end{multline}

\begin{proof}
    Follows immediately from Theorem 1 by substituting $r=1$ in formulas (\ref{dactom}), (\ref{eq:mid23}), (\ref{eq:lateralom}).
\end{proof}
\end{proposition}

\begin{proposition}

For a triple of positive integers $ \underline{\lambda}  = (m_1, m_2, m_3)$  branch points of the 2-fold Gaudin covering on $(V_{\underline{\lambda}}) _{m_1+m_2+m_3-2}^{\text{sing}}$ are given by the following equation 
$$
(m_1+m_3)^2u^2 - (  (m_1+m_3)^2  - (m_2+m_3)^2 +  (m_1+m_2)^2)u + (m_1+m_2)^2=0 $$ 

\begin{proof}
    Follows from a direct computation of discriminant of second-degree equation \eqref{eq:2-fold}
\end{proof}

\end{proposition}

\begin{hypothesis}

Fix a positive integer $r$. Suppose that all three positive integers $m_1,m_2,m_3$ tend to infinity in such a way that limits of pairwise ratios between them exist, and by abuse of notation, let $ \lim m_1 : m_2 : m_3 = M_1 : M_2 : M_3$. Then the branch points with positive (resp. negative) imaginary part of the Gaudin covering of $(V_{\underline{\lambda}}) _{m_1+m_2+m_3-2r}^{\text{sing}}$ go to the root with positive (resp. negative) imaginary part of the equation (see figures 5-6)
$$
(M_1+M_3)^2u^2 - (  (M_1+M_3)^2  - (M_2+M_3)^2   + (M_1+M_2)^2)u + (M_1+M_2)^2=0 $$

\end{hypothesis}

\begin{center}
\begin{minipage}{0.3\linewidth}
\includegraphics[width=\linewidth]{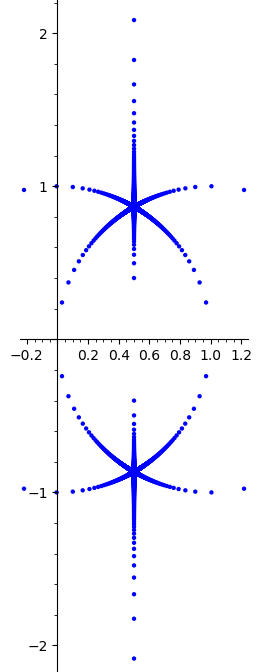}
\captionof{figure}{$M_i=1, r=3$}
\end{minipage}%
\hfill
\begin{minipage}{0.56\linewidth}
\includegraphics[width=\linewidth]{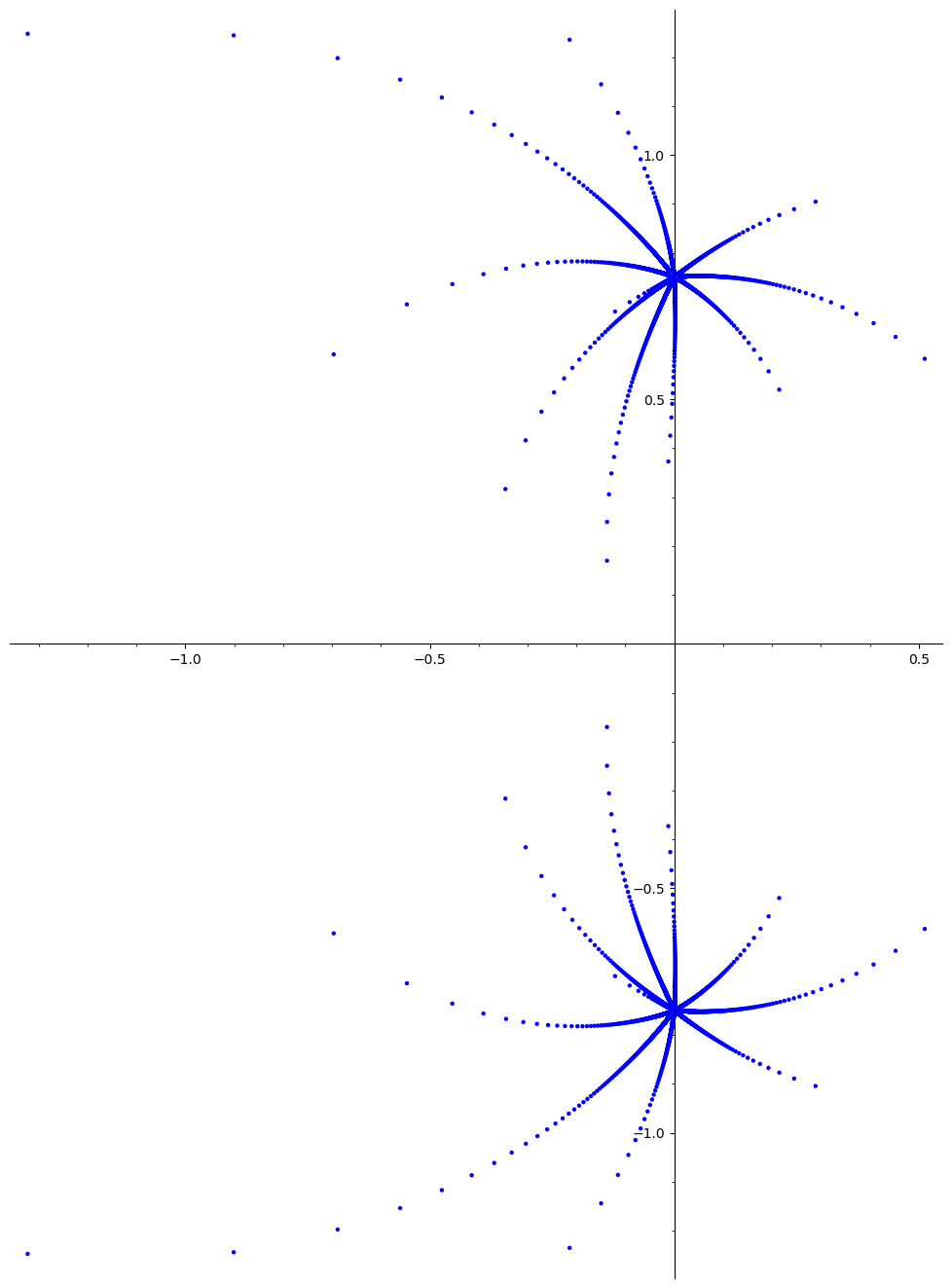}
\captionof{figure}{$M_1=1,M_2=2,M_3=3,r=4$}
\end{minipage}
\end{center}

\subsection{Ornaments for  $m_1, m_2, m_3$ much larger than $r$}


When $m_1, m_2, m_3$ are large compared to $r$, the branch points stack in rows and form a curvilinear triangle. In the context of branch points of algebraic curves similar patterns were observed in the work of Boris Shapiro and Milos Tater (see \cite{SHAPIRO} and \cite{SHAPIRO2}).

\begin{center}
    \includegraphics[height=12cm]{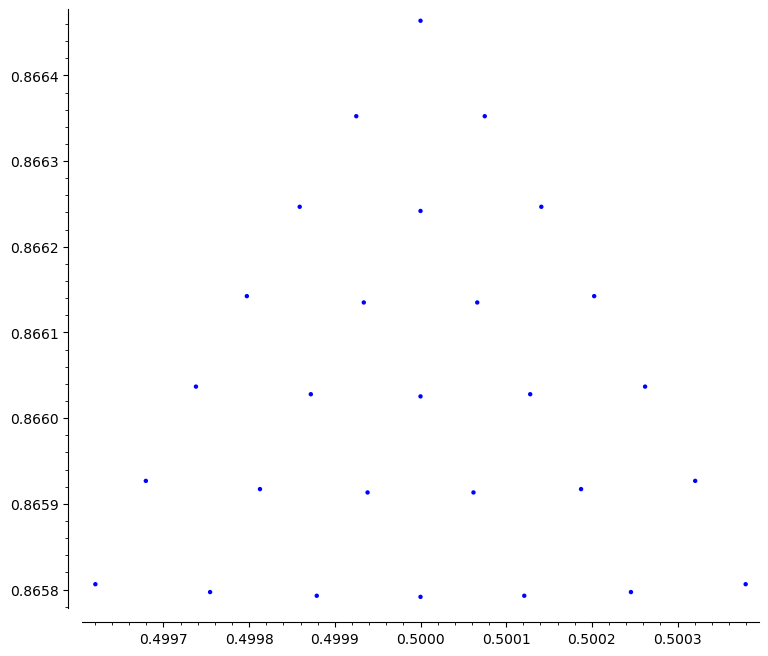}
\captionof{figure}{$m_i=10^6, r=7$}
\end{center}

\subsection{Description of monodromy of the Gaudin coverings}

When $u$ is real the spectrum is real and simple, so we can enumerate eigenvalues in the increasing order. Consider the following family of paths with superscripted permutations in the image of the monodromy action:

\begin{center}
    \includegraphics[height=13cm]{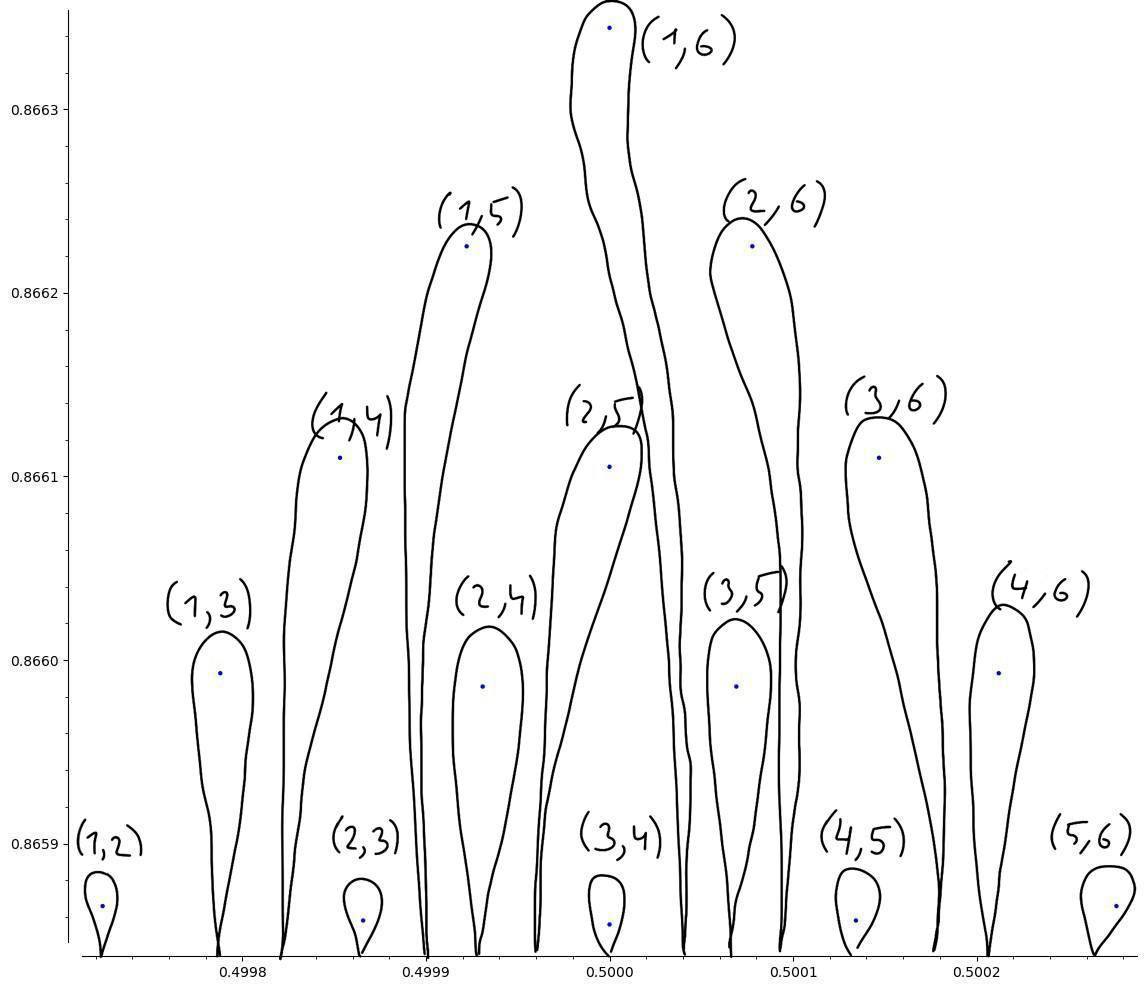}
\captionof{figure}{Monodromy with specific paths chosen}
\end{center}

Similar patterns arise in the work of Boris Shapiro and Milos Tater (see \cite{SHAPIRO} and \cite{SHAPIRO2}).

\section{Ornaments of branch points \label{Ornaments}}

In this section, we provide the concrete examples calculated by the computer program, which we wrote based on formulas (\ref{dactom}), (\ref{eq:mid13}), (\ref{eq:mid23}), (\ref{eq:lateralom}). 

\subsection{$m_1 = 3, m_2 = 4, m_3 = 4, r =4$}
Consider the Gaudin covering corresponding to the space $V_{3} \otimes V_{4} \otimes V_{4}$ and parameters $(0,1, \frac{1}{u},\infty )$.

The spectral curve of the subspace of singular vectors of weight $3$ (so that $r=4$) is given by the following equation \newline
$x^4 + (-10u - 10)x^3 + (-27u^2 + 162u - 27)x^2 + (360u^3 - 432u^2 - 432u + 360)x - 324u^4 - 864u^3 + 2376u^2 - 864u - 324$

It is smooth, its genus is equal to $3$. The branch points are the zeroes of the polynomial

 $u^{12} - 11364/1225 u^{11} + 56991/1225 u^{10} - 36208/245 u^9 + 9784071/30625 u^8 - 15151356/30625 u^7 + 3489114/6125 u^6 - 15151356/30625 u^5 + 9784071/30625 u^4 - 36208/245 u^3 + 56991/1225 u^2 - 11364/1225 u + 1$

They are simple. The monodromy group of the covering is $S_4$.

Here are the branch points in the picture: 

\begin{center}
\includegraphics[height=12cm]{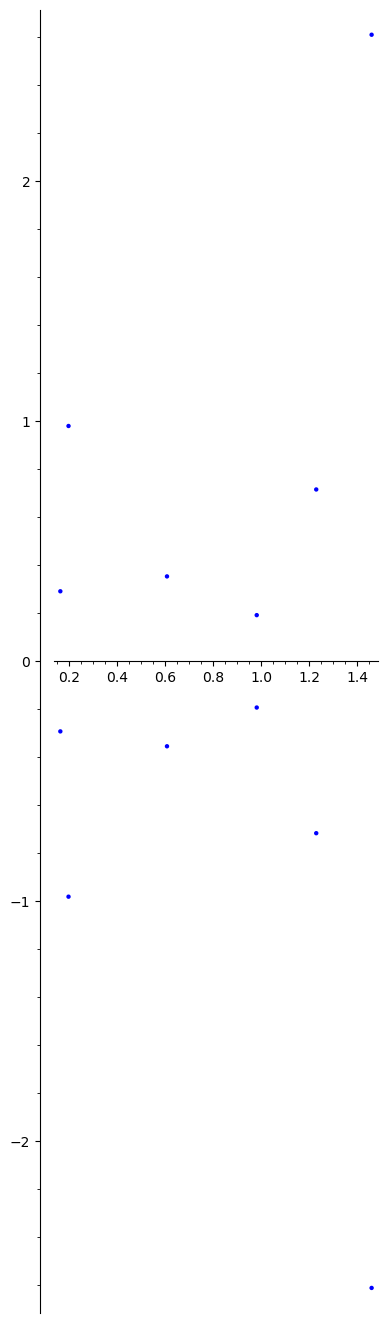}
\captionof{figure}{$m_1=3,m_2=m_3=r=4$}

\end{center}

\subsection{$m_1 = 10, m_2 = 10, m_3 =10$, $r = 6$}
Consider the Gaudin covering corresponding to the space $V_{10} \otimes V_{10} \otimes V_{10}$ and parameters $(0,1, \frac{1}{u},\infty )$.

The spectal curve of the subspace of singular vectors of weight $18$ (so that $r=6$) is given by the equation \newline $x^7 + (-3192 u^2 + 3192 u - 3192) x^5 + (18944 u^3 - 28416 u^2 - 28416 u + 18944) x^4 + (2455440 u^4 - 4910880 u^3 + 7366320 u^2 - 4910880 u + 2455440) x^3 + (-18777600 u^5 + 46944000 u^4 - 18777600 u^3 - 18777600 u^2 + 46944000 u - 18777600) x^2 + (-353376000 u^6 + 1060128000 u^5 - 2289600000 u^4 + 2812320000 u^3 - 2289600000 u^2 + 1060128000 u - 353376000) x + 1555200000 u^7 - 5443200000 u^6 + 6998400000 u^5 - 3888000000 u^4 - 3888000000 u^3 + 6998400000 u^2 - 5443200000 u + 1555200000$.

It is smooth, its genus is equal to $15$.

They are simple. The monodromy group of the covering is $S_7$.

Here are the branch points in the picture: 

\begin{center}
\includegraphics[height=12cm]{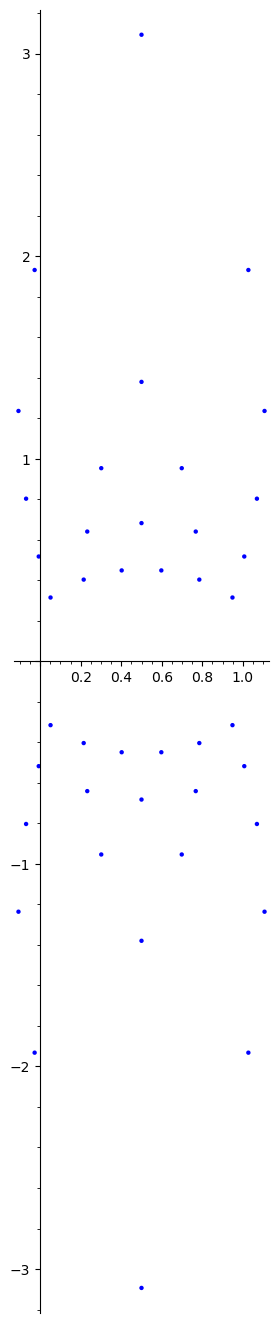}
\captionof{figure}{$m_1=m_2=m_3=10, r=6$}
\end{center}

\subsection{ $m_1 = 30, m_2 = 30, m_3 = 30, r = 20$}

Consider the Gaudin covering corresponding to the space $V_{30} \otimes V_{30} \otimes V_{30}$ and parameters $(0,1, \frac{1}{u},\infty )$.

The spectral curve of the subspace of singular vectors of weight $50$ ($r=20$) is smooth, its genus is $190$. All branch points are simple, and the monodromy group is $S_{21}$. 

Here are the branch points in the picture:

\begin{center}
\includegraphics[height=20cm]{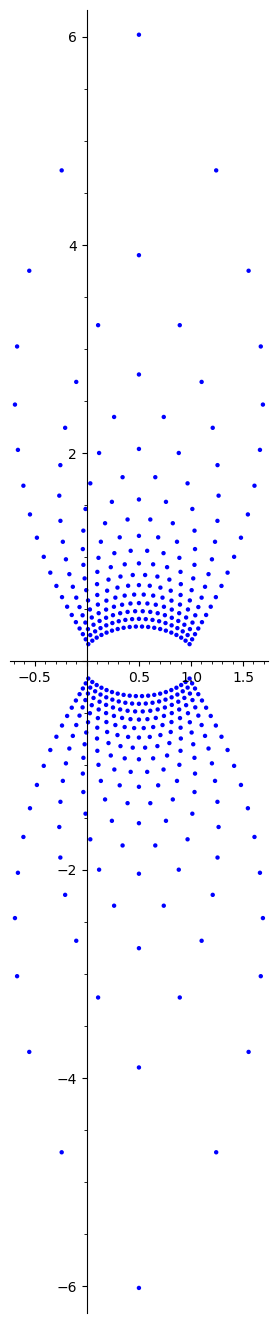}
\captionof{figure}{$m_1=m_2=m_3=30, r=20$}
\end{center}

\subsection{$m_1 = 31, m_2 = 32, m_3 = 33, r = 23$}

The Gaudin covering corresponding to the space $V_{31} \otimes V_{32} \otimes V_{33}$ and parameters $(0,1, \frac{1}{u},\infty )$:

The curve of the covering on the subspace of singular vectors of weight $49$ ($r=23$) is smooth, its genus is $253$. All branch points are simple, and the monodromy group is $S_{24}$. 
Here are the branch points in the picture:
\begin{center}
\includegraphics[height=22cm]{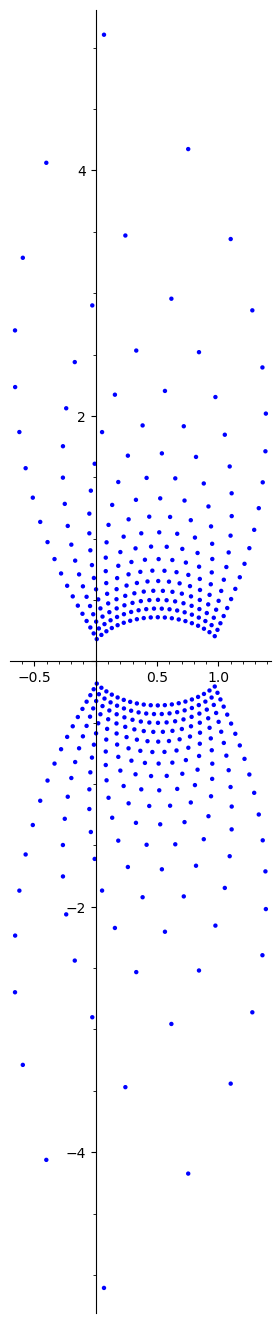}
\captionof{figure}{$m_1=31, m_2= 32, m_3=33, r=23$}

\end{center}

\subsection{Other ornaments of branch points}

\begin{center}
\begin{minipage}{0.49\linewidth}
\includegraphics[width=\linewidth]{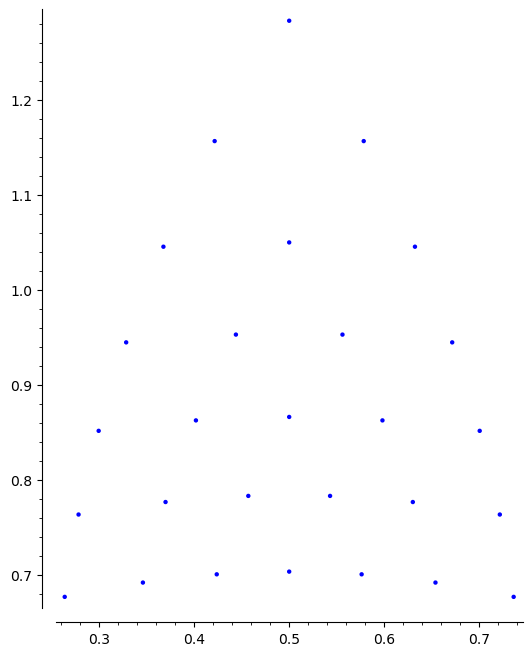}
\captionof{figure}{$m_i=50, r=7$}
\end{minipage}%
\hfill
\begin{minipage}{0.49\linewidth}
\includegraphics[width=\linewidth]{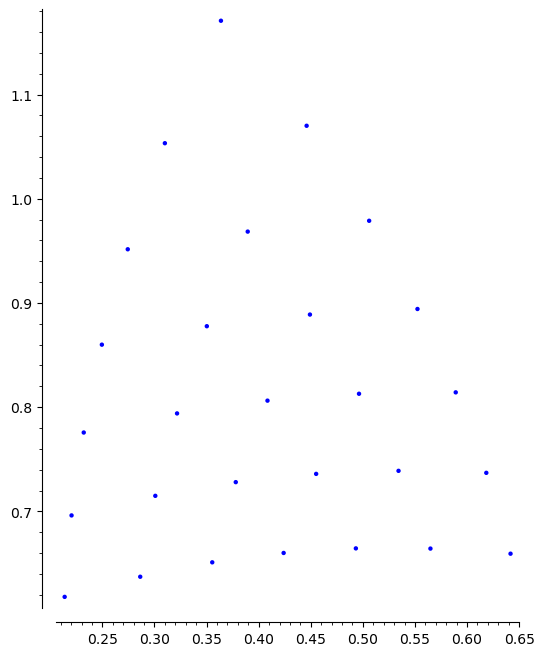}
\captionof{figure}{$m_1 = m_2 = 50, m_3=60,r=7$}
\end{minipage}
\end{center}

\begin{center}
\begin{minipage}{0.459\linewidth}
\includegraphics[width=\linewidth]{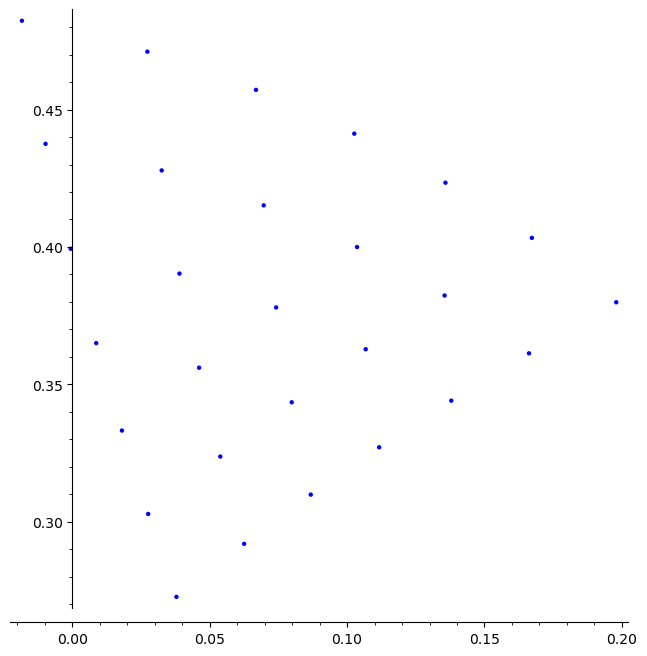}
\captionof{figure}{$m_1 = m_2 = 50, m_3=200,r=7$}
\end{minipage}%
\hfill
\begin{minipage}{0.541\linewidth}
\includegraphics[width=\linewidth]{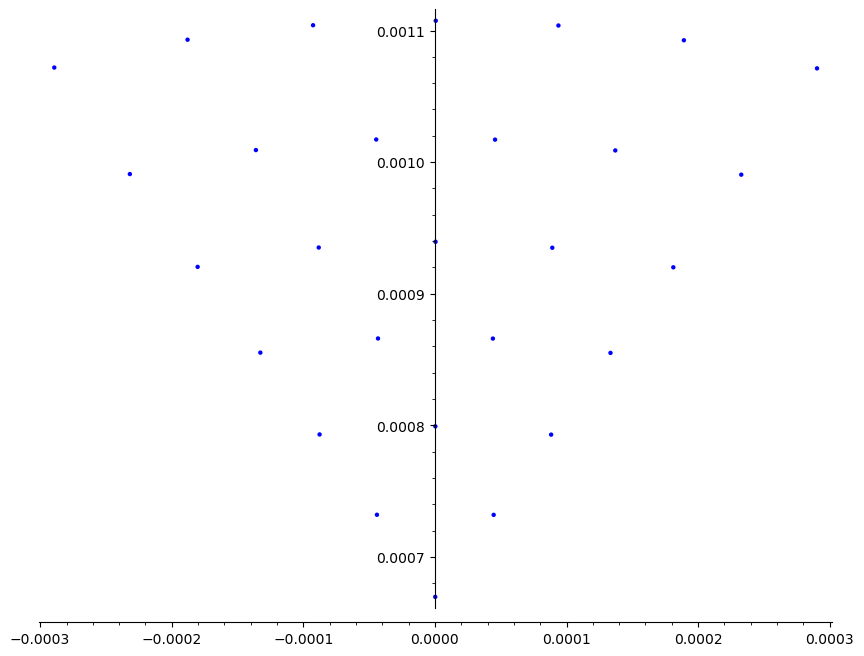}
\captionof{figure}{$m_1 = m_2 = 50, m_3=10^5,r=7$}
\end{minipage}
\end{center}

\begin{center}
\begin{minipage}{0.31\linewidth}
\includegraphics[width=\linewidth]{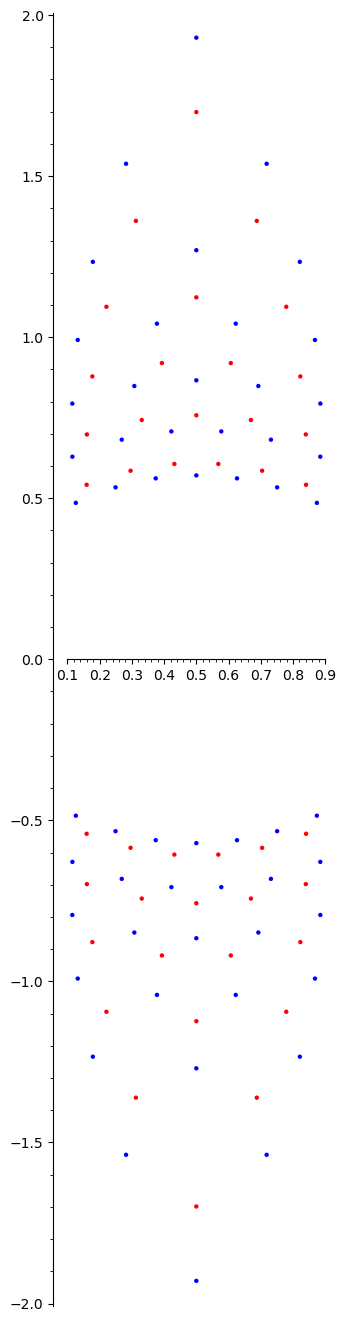}
\captionof{figure}{$m_i=20 ,r=6 \, (red), r=7 \,(blue)$}
\end{minipage}%
\hfill
\begin{minipage}{0.555\linewidth}
\includegraphics[width=\linewidth]{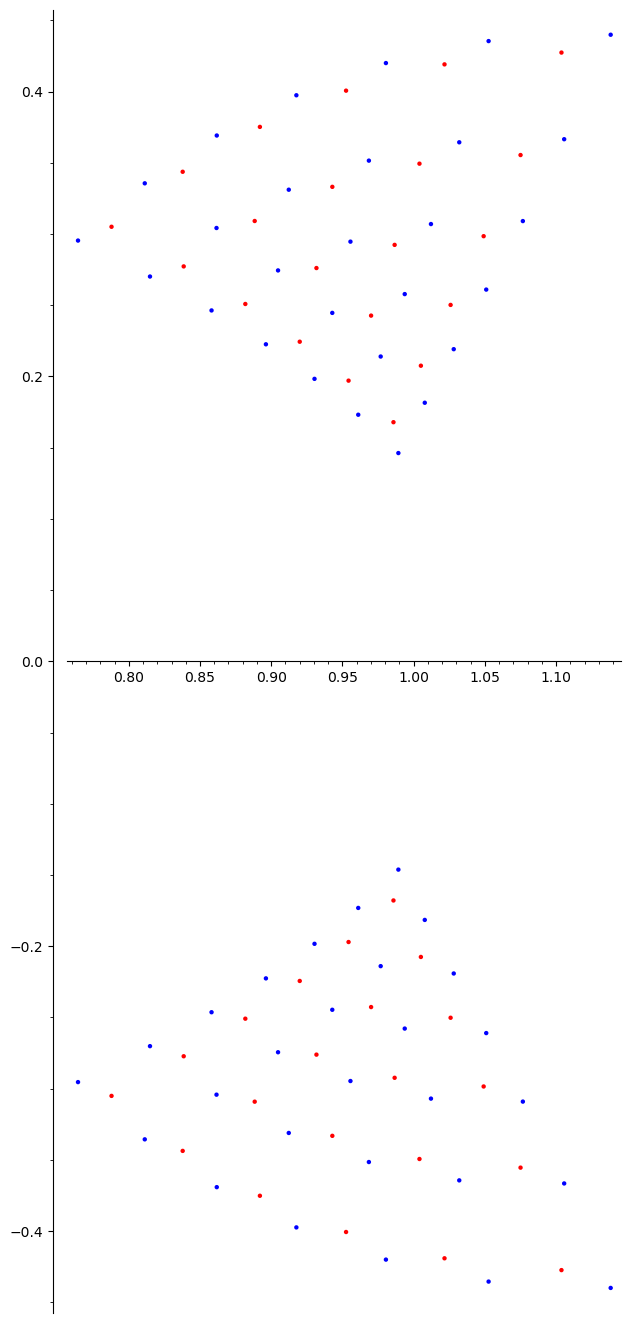}
\captionof{figure}{$m_1 = 100, m_2 = m_3 = 20, r=6 \,(blue), r= 7 \,(red)$}
\end{minipage}
\end{center}

\newpage 
\section{Acknowledgements}

The authors are grateful to Leonid Rybnikov, Yuri Chernyakov, Boris Shapiro, Vladimir Rubtsov, and Dmitry Vasiliev for useful discussions. The work of Natalia Amburg was funded by the Russian Science Foundation 
(Grant No. 24-12-00178). This work of Ilya Tosltukhin is an output of a research project implemented as part of the Basic Research Program at the National Research University Higher School of Economics (HSE University).

\vspace*{\fill}

Natalia Amburg:
\\
NRC "Kurchatov institute", Moscow
\\
Faculty of Mathematics, National Research University Higher School of  
Economics,  Moscow
\\
e-mail: amburg@mccme.ru

Ilya Tolstukhin:
\\
Faculty of Mathematics, National Research University Higher School of  
Economics,  Moscow
\\
The Skolkovo Institute of Science and Technology,  Moscow
\\
e-mail: itolstukhin@hse.ru


\begin{thebibliography}{9}
\bibitem{R1} 
L. Rybnikov, Cactus group and monodromy of Bethe vectors, International Mathematics Research Notices, Volume 2018, Issue 1, 3 January 2018, Pages 202-235

\bibitem{R2} 
L. Rybnikov, A proof of the Gaudin Bethe ansatz conjecture, International Mathematics Research Notices, 25 October 2018

\bibitem{HKRW} 
I. Halacheva, J. Kamnitzer, L.Rybnikov, A. Weeks, Crystals and monodromy of Bethe vectors, Duke Mathematical Journal. 2020. Vol. 169. No. 12. P. 2337-2419.


\bibitem{FFR}
B. Feigin, E. Frenkel, N. Reshetikhin, Gaudin Model, Bethe Ansatz and Critical Level, Comm. Math. Phys., 166 (1994), pp. 27 - 62

\bibitem{GD}
M. Gaudin, Diagonalisation d’une classe d’hamiltoniens de spin, J. de Physique, t.37, N 10, p. 1087–1098,
1976

\bibitem{Var}

I. Scherbak, A. Varchenko, Critical Points of the Product of Powers of Linear Functions and Families of Bases of Singular Vectors, Compositio Mathematica, Tome 97 (1995) no. 3, pp. 385-401.

\bibitem{Scher1}
I. Scherbak, A. Varchenko, Critical Points of Functions, sl2 Representations, and Fuchsian Differential Equations with only Univalued Solutions, January 2002, Moscow Mathematical Journal 

\bibitem{Scher}
I. Scherbak, Gaudin's Model and the Generating Function of the Wronski map, 30 August 2003, Mathematics Banach Center Publications


\bibitem{Muk}

E. Mukhin, V. Tarasov, A. Varchenko, Schubert calculus and representation of the general linear group, Journal of the American Mathematical Society, Volume 22, Number 4, October 2009, pp. 909–940



\bibitem{SHAPIRO}
B. Shapiro, M. Tater,  Asymptotics and monodromy of the algebraic spectrum of quasi-exactly solvable sextic oscillator, Experimental Mathematics, Volume 28, 2019 - Issue 1, pp. 16-23

\bibitem{SHAPIRO2} B. Shapiro, M. Tater, On spectral asymptotic of quasi-exactly solvable quartic potential, Analysis and Mathematical Physics volume 12, Article number: 2 (2022)


\bibitem{Azad}
A. Saifullin, Generalized Jucys–Murphy elements, B.A. thesis 





\end{thebibliography}
\end{document}